\documentclass[preprint,3p,times]{elsarticle}

\usepackage{amssymb,amsmath,amsfonts,amssymb,bbm,amsthm, bm}
\usepackage{algorithm,algorithmic}
\usepackage{booktabs} 
\usepackage{graphicx} 
\usepackage{natbib}

\newtheorem{theorem}{Theorem}
\newtheorem{lemma}[theorem]{Lemma}

\newtheorem{corollary}[theorem]{Corollary}
\newtheorem{proposition}[theorem]{Proposition}

\theoremstyle{definition}
\newtheorem{definition}[theorem]{Definition}

\newtheorem{instance}{Instance}
\newtheorem{openP}{Open Problem}
\newtheorem{problem}{Problem}

\floatname{algorithm}{Mechanism}

\usepackage{xcolor} 

\newcommand{\M}{\mathcal{M}}

\journal{Artificial Intelligence}

\begin{document}

\begin{frontmatter}



\title{On Existence of Truthful Fair Cake Cutting Mechanisms}


\author[1]{Xiaolin Bu}
\author[2]{Jiaxin Song}
\author[3]{Biaoshuai Tao}

\address[1]{Shanghai Jiao Tong University, lin\_bu@sjtu.edu.cn}
\address[2]{Shanghai Jiao Tong University, sjtu\_xiaosong@sjtu.edu.cn}
\address[3]{Shanghai Jiao Tong University, bstao@sjtu.edu.cn}

\begin{abstract}

 We study the fair division problem on divisible heterogeneous resources (the cake cutting problem) with strategic agents, where each agent can manipulate his/her private valuation to receive a better allocation.
    A (direct-revelation) mechanism takes agents' reported valuations as input and outputs an allocation that satisfies a given fairness requirement.
    A natural and fundamental open problem, first raised by Chen, Lai, Parkes, and Procaccia~\cite{CL10} and subsequently raised in reference~\cite{Pro13,aziz2014cake,branzei2015dictatorship,menon2017deterministic,bei2017cake,bei2020truthful}, etc., is whether there exists a deterministic, truthful, and envy-free (or even proportional) cake cutting mechanism.
    In this paper, we resolve this open problem by proving that there does not exist a deterministic, truthful and proportional cake cutting mechanism, even in the special case where all of the following hold:
    \begin{itemize}
    \item there are only two agents;
    \item each agent's valuation is a piecewise-constant function;
    \item each agent is hungry: each agent has a strictly positive value on any part of the cake.
    \end{itemize}
    The impossibility result extends to the case where the mechanism is allowed to leave some part of the cake unallocated.
    
    We also present a truthful and envy-free mechanism when each agent's valuation is piecewise-constant and monotone.
    However, if we require Pareto-optimality, we show that truthful is incompatible with approximate proportionality for any positive approximation ratio even for piecewise-constant and monotone value density functions.
    
    To circumvent the main impossibility result, we aim to design mechanisms that possess a certain degree of truthfulness.
    Motivated by the kind of truthfulness possessed by the classical I-cut-you-choose protocol, we propose a weaker notion of truthfulness, \emph{the proportional risk-averse truthfulness}.
    We show that the well-known moving-knife (Dubins-Spanier) procedure and Even-Paz algorithm do not have this truthful property.
    We propose a mechanism that is proportionally risk-averse truthful and envy-free, and a mechanism that is proportionally risk-averse truthful that always outputs allocations with connected pieces.
\end{abstract}




\end{frontmatter}

\noindent\textbf{Note: A preliminary version of this paper is published in EC'22. New results are added in this version, and are presented in Sect.~\ref{sect:monotone}, \ref{append:nom}, and \ref{append:missingproofs}.}

\section{Introduction}
\label{sect:intro}
The cake cutting problem studies the allocation of a piece of divisible heterogeneous resource to multiple agents, normally with a given fairness requirement.
The cake is a metaphor for divisible heterogeneous resources, which is normally modeled as an interval $[0,1]$.
Different agents have different valuations on different parts of the interval.
Typically, each agent's valuation is described by a \emph{value density function} $f:[0,1]\to\mathbb{R}_{\geq0}$, and his/her value on a subset $X\subseteq[0,1]$ is given by the Riemann integral $\int_Xf(x)dx$.
Starting with Steinhaus~\cite{Steinhaus1948}, the cake cutting problem has been widely studied by mathematicians (e.g., \cite{dubins1961cut,even1984note,brams1995envy,edward1999rental,brams2006better}), economists (e.g., \cite{thomson1983problems,thomson1983fair,abdulkadirouglu2004room}), and computer scientists (e.g., most of the papers cited by this paper).
See the books~\cite{brams1996fair}, \cite{robertson1998cake} and Part II of the book~\cite{brandt2016handbook} and the survey~\cite{Pro13}.

Two of the most widely studied fairness criteria are \emph{proportionality} and \emph{envy-freeness}.
An allocation is proportional if each agent believes (s)he receives a share with a value that is at least a $\frac1n$ fraction of the value of the entire cake (where $n$ is the number of the agents).
An allocation is envy-free if each agent believes (s)he receives a share that has weakly more value than the share allocated to each of the other agents (i.e., an agent does not envy any other agents).
Formal definitions for the two notions are in Sect.~\ref{sect:prelim}.
If we require that the entire cake needs to be allocated (i.e., discarding some part of the cake is disallowed), an envy-free allocation is always proportional.
It is well-known that envy-free allocations (with the entire cake allocated) always exist~\cite{brams1995envy}, even if we require each agent must receive a connected interval~\cite{edward1999rental}.
In addition to the existence, the algorithm design aspect has also been considered in a long history~\cite{dubins1961cut,even1984note,stromquist2008envy,aziz2015discrete,aziz2016discrete}.
In particular, we know how to compute a proportional allocation~\cite{dubins1961cut,even1984note} and an envy-free allocation~\cite{aziz2016discrete} for any number of agents.

However, a fundamental issue when deploying a certain cake cutting algorithm is that agents are self-interested and may manipulate and misreport their valuations to the algorithm to get better allocations.
This motivates the study of the cake-cutting problem from a game-theoretical aspect, in particular, a mechanism design aspect.
Is there a \emph{truthful} and fair cake cutting mechanism such that truth-telling is each agent's dominant strategy?
This question was first proposed by Chen, Lai, Parkes, and Procaccia~\cite{CL10}.

To answer this question, we first need to address the following issue: how can we represent a value density function succinctly?
Two different approaches have been considered in the past literature.
In the first approach (e.g., \cite{brams1995envy,robertson1998cake,stromquist2008envy,kurokawa2013cut,aziz2015discrete,aziz2016discrete}), the mechanism communicates with the agents by a query model called \emph{the Robertson-Webb query model}, where the mechanism learns the valuation of each agent through a sequence of queries that are of the following two types:
\begin{itemize}
    \item $\textbf{Eval}_i(x,y)$: ask agent $i$ his/her value on the interval $[x,y]$;
    \item $\textbf{Cut}_i(x,r)$: ask agent $i$ for a point $y$ where $[x,y]$ is worth exactly $r$.
\end{itemize}
In the second approach (e.g., \cite{mossel2010truthful,CL10,bei2012optimal,menon2017deterministic,bei2017cake,bei2020truthful}), the value density function is assumed to be \emph{piecewise-constant}.
Piecewise-constant functions can approximate most natural real functions arbitrarily closely, and they can be succinctly encoded.
The mechanism then takes the $n$ encoded value density functions as input and outputs an allocation.
These mechanisms are called \emph{direct revelation} mechanisms.

In the setting with the Robertson-Webb query model, the game agents are playing is an \emph{extensive-form game}, whereas, in the piecewise-constant valuation setting, this is a one-round game where all the agents report their valuations simultaneously.
Naturally, when truthfulness is concerned, agents in the first setting have much more room for manipulation.
Indeed, for the first setting, Kurokawa, Lai, and Procaccia~\cite{kurokawa2013cut} prove that no truthful and envy-free mechanism terminates within a bounded number of Robertson-Webb queries.
A strong impossibility result by Br{\^a}nzei and Miltersen~\cite{branzei2015dictatorship} show that, for any truthful mechanism, there exists an agent who receives a zero value.
In particular, when there are only two agents, the only truthful mechanism is essentially the one that allocates the entire cake to a single agent.

For direct revelation mechanisms, Chen, Lai, Parkes, and Procaccia~\cite{CL10} give the first truthful envy-free cake cutting mechanism that works when each agent's valuation is \emph{piecewise-uniform}, a special case of piecewise-constant valuations with the additional assumption that each value density function takes value either $0$ or $1$.
Chen, Lai, Parkes, and Procaccia~\cite{CL10} then propose the following natural open problem.
\begin{problem}\label{openP:main}
Does there exist a (deterministic) truthful, envy-free (or even proportional) cake cutting mechanism for piecewise-constant value density functions?
\end{problem}
Many researchers make partial progress on this problem in the past decade.
Aziz and Ye~\cite{aziz2014cake} show that there exists no truthful mechanism that satisfies either one of the following properties:
\begin{itemize}
    \item Proportional and Pareto-optimal;
    \item Robust-proportional and non-wasteful (non-wasteful means that no piece is allocated to an agent who does not want it, a notion weaker than Pareto-optimality).
\end{itemize}
Menon and Larson~\cite{menon2017deterministic} show that there exists no truthful mechanism that is even approximately proportional, with the constraint that each agent must receive a connected piece.
Bei, Chen, Huzhang, Tao, and Wu~\cite{bei2017cake} show that there exists no truthful, proportional mechanism under any one of the following three settings:
\begin{itemize}
    \item the mechanism is non-wasteful;
    \item the mechanism is position-oblivious (meaning that the allocation of a cake-part is based only on the agents' valuations of that part, and not on its relative position on the cake);
    \item agents report the value density functions sequentially, where an agent's strategy can depend on the reports of the previous agents.
\end{itemize}

On the positive side, the mechanism proposed by Chen, Lai, Parkes, and Procaccia~\cite{CL10} for piecewise-uniform value density functions is further studied by Maya and Nisan~\cite{maya2012incentive} and Li, Zhang, and Zhang~\cite{li2015truthful}.
Maya and Nisan~\cite{maya2012incentive} characterize truthful mechanisms and show that the mechanism proposed in reference~\cite{CL10} is unique in some sense.
Li, Zhang, and Zhang~\cite{li2015truthful} show that this mechanism also works in the setting where agents have externalities.
Bei, Huzhang, and Suksompong~\cite{bei2020truthful} propose a truthful envy-free mechanism for piecewise-uniform value density functions that do not need the \emph{free-disposal} assumption, an assumption made in the mechanism in reference~\cite{CL10}.
Designing truthful and fair allocations has also been studied for value density functions that are more restrictive than piecewise-uniform~\cite{alijani2017envy,seddighin2019expand,asano2020cake}.
As can be seen above, most of the positive results are regarding piecewise-uniform valuations or even more restrictive ones.

Despite the above-mentioned progress, Problem~\ref{openP:main} remains open.

All the mechanisms mentioned above are deterministic.
If we allow randomized mechanisms, a simple mechanism proposed by Mossel and Tamuz~\cite{mossel2010truthful} is universal envy-free and truthful in expectation.
However, randomized mechanisms have many drawbacks.
Firstly, agents can be risk-seeking or risk-averse and may have different views on a truthful-in-expectation randomized mechanism.
Secondly, agents may have concerns about the source of the randomness.
It is costly to find a trustworthy random source.
Agents receiving less utility due to randomness may believe they have not been treated fairly.

\subsection{Our Results}
As the main result of this paper, we resolve Problem~\ref{openP:main} by proving that there does not exist a (deterministic) truthful proportional cake cutting mechanism.
This impossibility result can be extended to the setting where there are only two agents, each agent has a strictly positive value on any part of the cake (we say that the agents are \emph{hungry} in this case), and the mechanism is allowed to leave some part of the cake unallocated.
We further show that the impossibility result extends to the setting where only approximate proportionality is required, for some constant approximation ratio sufficiently close to $1$.

\paragraph{Main Result} \textit{There does not exist a deterministic, truthful, and (approximately) proportional mechanism, even if there are only two agents, agents are hungry, and the mechanism is allowed to discard some parts of the cake.} (Theorem~\ref{thm:main2} and Theorem~\ref{thm:main2_approx})

We next consider a natural special case where all agents' value density functions are monotone, in addition to being piecewise-constant.
We show that truthful is compatible with envy-freeness under this setting.
However, truthful and fairness are incompatible with Pareto-optimality.

\paragraph{Result 2}
\textit{There exists a truthful and envy-free mechanism for piecewise-constant and monotone value density functions.} (Theorem~\ref{thm:monotone})

Since our mechanism allocates the entire cake without discarding any part, this mechanism is also proportional.

\paragraph{Result 3}
\textit{There does not exist a truthful, approximately proportional, and Pareto-optimal mechanism for piecewise-constant and monotone value density functions. This is true for any positive approximation factor on proportionality.} (Theorem~\ref{thm:monotonePO})

To circumvent the main impossibility result, we propose a weaker truthful notion called \emph{risk-averse truthful}.
This is motivated by the truthful guarantee of the \emph{I-cut-you-choose} protocol (the protocol is defined in Sect.~\ref{sect:RAT}, after Theorem~\ref{thm:Nash}).
Our risk-averse truthful notion captures the risk-averseness of the agents and the setting where an agent does not know other agents' valuations.
Informally, a mechanism is risk-averse truthful if either each agent's misreporting of his/her valuation is not beneficial, or there is a possibility that the misreporting will hurt the agent's utility (see Definition~\ref{def:wrat}).
Based on the solution concept of proportionality, we also consider a truthful notion called \emph{proportionally risk-averse truthful} that is stronger than risk-averse truthful.
A proportional mechanism is proportionally risk-averse truthful if either each agent's misreporting of his/her valuation is not more beneficial, or there is a possibility that the misreporting will make the agent even fail to get a proportional allocation (see Definition~\ref{def:rat}).

We show that those well-known algorithms, e.g., the moving-knife procedure~\cite{dubins1961cut} and the Even-Paz algorithm~\cite{even1984note}, do not satisfy this truthful property.
We then propose a mechanism that is proportionally risk-averse truthful and envy-free, and a mechanism that is proportionally risk-averse truthful that always outputs allocations with connected pieces.

\paragraph{Result 4} \textit{There exists a mechanism that is proportionally risk-averse truthful and envy-free.} (Theorem~\ref{thm:rat_ef})

\paragraph{Result 5} \textit{There exists a mechanism that is proportionally risk-averse truthful that always outputs allocations with connected pieces.} (Theorem~\ref{thm:rat_p_proportional}, Theorem~\ref{thm:rat_p_risk-averse_hungry} and Theorem~\ref{thm:rat_p_change})

\medskip

Our risk-averse truthful notion is similar but stronger than the truthful notion defined by Brams, Jones, and Klamler~\cite{brams2006better}.
They also consider the setting where each agent does not know the valuations of the other agents, and, in their notion, a mechanism is truthful if each agent cannot misreport his/her valuation and ``assuredly'' do better.
It is possible that misreporting will always be no harm, sometimes make the agent's utility unchanged, and sometimes be beneficial.
In this case, the misreporting cannot ``assuredly do better''.
It satisfies the truthful notion in reference~\cite{brams2006better} but not our risk-averse truthfulness.
For example, the above-mentioned moving-knife procedure satisfied the truthful notion in reference~\cite{brams2006better} but not our risk-averse truthfulness.
See Sect.~\ref{sect:RAT} and \ref{append:Brams} for details and more comparisons.

Another similar truthful notion that takes into account agents' uncertainty about other agents' utilities, called \emph{not obvious manipulability}, is proposed by Troyan and Morrill~\cite{troyan2020obvious}.
It requires that misreporting the utility function is non-beneficial in both the worst case and the best case.
Besides many technical differences, Troyan and Morrill's notion is also conceptually different from ours.
Our notions, as well as Brams et al.'s, focus more on agents' risk-averse motivation, whereas Troyan and Morrill's notion mainly captures the difficulty of finding a beneficial deviation.
See Sect.~\ref{sect:RAT} and \ref{append:nom} for details and more comparisons.

\subsection{Related Work}
In the previous part, we have discussed many related work about the truthfulness in the cake-cutting problem.
In this section, we will mainly go through some related work about the truthfulness in other fair division settings.

\paragraph{Truthfulness under homogeneous divisible items setting}
Many work~\cite{cole2013mechanism,guo2010strategy,han2011strategy,cole2013positive,zivan2010reducing} considers the allocation of multiple homogeneous divisible items.
This model is similar to our cake-cutting problem with piecewise-constant value density functions.
Indeed, for piecewise-constant value density functions, we can partition the cake $[0,1]$ to many intervals where all the functions are constant on each of these intervals.
Each such interval can then be viewed as a divisible homogeneous item.
However, if we are dealing with truthfulness, the two models are fundamentally different: in the cake-cutting setting, an agent can misreporting the value density function by changing the set of the discontinuity points, which will change the ``definitions of the items''.
In fact, under the homogeneous divisible items setting, allocating each item evenly to all the agents is a trivial truthful and envy-free (and proportional) mechanism.
However, the main result in this paper shows that truthfulness and proportionality are incompatible in the cake-cutting setting.

Since truthfulness and fairness can trivially be guaranteed for the homogeneous divisible items setting, it is tempting to include efficiency into consideration.
Guo and Conitzer~\cite{guo2010strategy}, Han et al.~\cite{han2011strategy} and Cole et al.~\cite{cole2013positive} present many results on the upper bound and the lower bound of the approximation ratio on the \emph{social welfare} for truthful mechanisms.
Cole et al.~\cite{cole2013mechanism} provide a truthful mechanism where each agent can receive at least a $1/e$ fraction of his/her value in a \emph{proportionally fair} allocation\footnote{An allocation is \emph{proportionally fair} if it maximizes the \emph{Nash social welfare}---the product of all the agents' utilities. It is widely known that an agent's utility is the same in all proportionally fair allocations.}.
Zivan et al.~\cite{zivan2010reducing} focus on achieving envy-freeness and Pareto-optimality while reducing, but not eliminating, the incentive to misreport.

\paragraph{Truthfulness under indivisible items setting}
For indivisible items, even weak fairness notions cannot be achieved when truthfulness is enforced.
Even for two agents, Caragiannis et al.~\cite{caragiannis2009low} and Amanatidis et al.~\cite{amanatidis2016truthful} present strong impossibility results. 
Amanatidis et al.~\cite{amanatidis2017truthful} also characterize the deterministic truthful mechanisms for two agents with additive valuations, and show that no reasonable fairness notion can be guaranteed.
Garg and Psomas~\cite{garg2022efficient} show that, even when allowing randomized mechanisms, the only truthful and Pareto-optimal mechanism is a serial dictatorship.
Only for some special cases, truthfulness and fairness can be compatible~\cite{bogomolnaia2004random, halpern2020fair, amanatidis2021maximum, babaioff2021fair}. 
Psomas and Verma~\cite{psomas2022fair} consider the previously mentioned relaxed truthful notion ``not obvious manipulability'' and show that it is compatible with the fairness notion \emph{envy-free up to one item} (a natural relaxation of envy-freeness in the setting with indivisible items) and Pareto-optimality.

\paragraph{Truthfulness in other models}
Apart from the above two models, truthful is also widely studied in other scenarios such as house allocation and stable matching~\cite{zhou1990conjecture, abdulkadirouglu1999house, miyagawa2002strategy, krysta2014size}.
However, instead of considering fairness notions, these papers mainly focus on stability or efficiency. 
Some other researches are about truthful random assignments~\cite{katta2006solution, aziz2017impossibilities, mennle2021partial}, where randomized mechanisms satisfying \emph{truthfulness in expectation} are considered.

\paragraph{Other aspects of cake-cutting}
There are also many other papers studying the cake-cutting problem without the strategic aspect.
They focus on other aspects such as computational complexity~\cite{stromquist2008envy,edmonds2006cake,procaccia2009thou,edmonds2006balanced,deng2012algorithmic,branzei2017query,goldberg2020contiguous,seddighin2019expand} and economic efficiency~\cite{bei2012optimal,cohler2011optimal,brams2012maxsum,aumann2012computing,bertsimas2011price,caragiannis2012efficiency,aumann2015efficiency,arzi2011throw} that are not discussed in this paper.

\subsection{Structure of This Paper}
In Sect.~\ref{sect:prelim}, we formally describe the model of the cake cutting problem with direct revelation mechanisms.
In Sect.~\ref{sect:impossibility}, we present our main result: resolving Problem~\ref{openP:main} and extending the impossibility result to the approximation setting.
In Sect.~\ref{sect:monotone}, we consider monotone valuations, present a truthful and envy-free mechanism, and show that truthful and fairness cannot be compatible with Pareto-optimality.
Sect.~\ref{sect:RAT} to Sect.~\ref{sect:RAT_P} discuss the relaxations on dominant strategy truthfulness and present several mechanisms that satisfy the relaxed truthful notions.
We conclude our paper and discuss some future research directions in Sect.~\ref{sect:conclusion}.

\section{Preliminaries}
\label{sect:prelim}
The cake is modeled as the interval $[0,1]$, which is allocated to $n$ agents.
Each agent $i$ has a \emph{value density function} $f_i:[0,1]\to\mathbb{R}_{\geq0}$ that describes his/her preference on the cake.
A value density function $f_i$ is \emph{piecewise-constant} if $[0,1]$ can be partitioned into finitely many intervals, and $f_i$ is constant on each of these intervals.
We will assume agents' value density functions are piecewise-constant throughout the paper, although our results in Sect.~\ref{sect:RAT_P} do not rely on this.
Agent $i$ is \emph{hungry} if $f_i(x)>0$ for any $x\in[0,1]$.
Given a subset $X\subseteq[0,1]$, agent $i$'s \emph{utility} on $X$, denoted by $v_i(X)$, is given by $$v_i(X)=\int_Xf_i(x)dx.$$

An allocation $(A_1,\ldots,A_n)$ is a collection of mutually disjoint subsets of $[0,1]$, where $A_i$ is the subset allocated to agent $i$.
An allocation is \emph{entire} if $\bigcup_{i=1}^nA_i=[0,1]$.
Notice that an impossibility result without the entire requirement is stronger than an impossibility result with this requirement.
An allocation is \emph{proportional} if each agent receives his/her average share of the entire cake:
$$\forall i:\quad v_i(A_i)\geq\frac1n v_i([0,1]).$$
An allocation is \emph{$\alpha$-approximately proportional} if $\frac1n$ above is changed to $\frac\alpha{n}$.
An allocation is \emph{envy-free} if each agent receives a portion that has a weakly higher value than any portion received by any other agent, based on his/her own valuation:
$$\forall i,j:\quad v_i(A_i)\geq v_i(A_j).$$
An entire envy-free allocation is always proportional.
In the case of two agents, if an allocation is entire, it is envy-free if and only if it is proportional.
In Sect.~\ref{sect:RAT_P}, we consider a specific kind of allocations where each agent needs to receive a connected piece of cake, i.e., each $A_i$ is an \emph{interval}.

A \emph{mechanism} is a function $\M$ that maps $n$ value density functions $F=(f_1,\ldots,f_n)$ to an allocation $(A_1,\ldots,A_n)$.
Given $\M(F)=(A_1,\ldots,A_n)$, we write $\M_i(F)=A_i$.
That is, $\M_i(F)$ outputs the share allocated to agent $i$, given input $F=(f_1,\ldots,f_n)$.
A mechanism is proportional/envy-free if it always outputs a proportional/envy-free allocation with respect to the input $F=(f_1,\ldots,f_n)$.
A mechanism is entire if it always outputs entire allocations.
In this paper, we consider only deterministic mechanisms.

A mechanism $\M$ is \emph{truthful} if each agent's dominant strategy is to report his/her true value density function.
That is, for each $i\in[n]$, any $(f_1,\ldots,f_n)$ and any $f_i'$,
$$v_i\left(\M_i(f_1,\ldots,f_n)\right)\geq v_i\left(\M_i(f_1,\ldots,f_{i-1},f_i',f_{i+1},\ldots,f_n)\right).$$

As a clarification, when proportionality/envy-freeness is concerned, a mechanism must output an allocation that is proportional/envy-free with respect to the \emph{reported value density functions}; when truthfulness is concerned, we require each agent's misreporting does not give this agent strictly more utility, and the utility here is with respect to this agent's \emph{true value density function}.

\section{Impossibility Result for Truthful Proportional Mechanism}
\label{sect:impossibility}
In this section, we prove the following theorem.
\begin{theorem}\label{thm:main2}
There does not exist a truthful proportional mechanism, even when all of the following hold:
\begin{itemize}
    \item there are two agents;
    \item each agent's value density function is piecewise-constant;
    \item each agent is hungry: each $f_i$ satisfies $f_i(x)>0$ for any $x\in[0,1]$;
    \item the mechanism needs not to be entire: the mechanism may throw away parts of the cake.
\end{itemize}
\end{theorem}

We will prove Theorem~\ref{thm:main2} by contradiction.
Suppose there exists a truthful proportional mechanism $\M$ for two agents.
For a description of the main idea behind the proof, we construct multiple cake cutting instances, analyze the outputs of $\M$ on these instances, and prove that truthfulness and proportionality cannot be guaranteed on all these instances.
In particular, we will construct six instances.
For the first five instances, we show that the outputs of $\M$ are unique.
Based on the outputs for the first five instances, we show that any allocation output by $\M$ for the sixth instance will violate either proportionality or truthfulness.
The six instances constructed are shown in Table~\ref{tab}.

\begin{table}
    \centering
    \begin{tabular}{|c|c|}
    \hline
    Instance & Allocation \\
    \hline
      \includegraphics{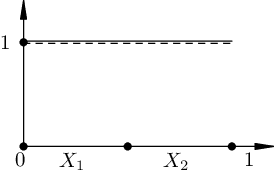}   & $\M(F^{(1)})=(X_1,X_2)$ \\
      \hline
      \includegraphics{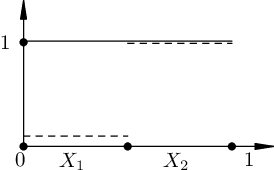}  & $\M(F^{(2)})=(X_1,X_2)$\\
      \hline
      \includegraphics{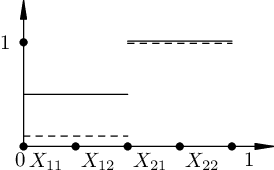}  & $\M(F^{(3)})=(X_{11}\cup X_{21},X_{12}\cup X_{22})$\\
      \hline
      \includegraphics{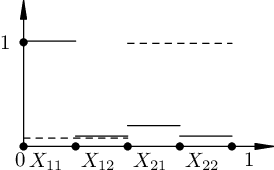}  & $\M(F^{(4)})=(X_{11}\cup X_{21},X_{12}\cup X_{22})$\\
      \hline
      \includegraphics{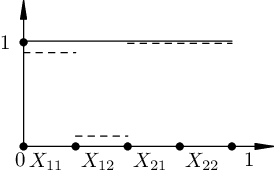}  & $\M(F^{(5)})=(X_1,X_2)$\\
      \hline
      \includegraphics{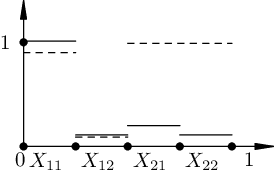}  & See Sect.~\ref{sect:submain}\\
      \hline
    \end{tabular}
    \caption{Instances constructed for the proof of Theorem~\ref{thm:main2} and the corresponding allocations given by $\M$. The value density for agent $1$ is shown in solid lines, and the value density for agent $2$ is shown in dashed lines.}
    \label{tab}
\end{table}

We start with the simplest cake cutting instance.
\begin{instance}\label{ins1}
$F^{(1)}=(f_1^{(1)},f_2^{(1)})$, where $f_1^{(1)}(x)=1$ and $f_2^{(1)}(x)=1$ for $x\in[0,1]$.
\end{instance}

To ensure proportionality, we must have $|\M_1(F^{(1)})|=|\M_2(F^{(1)})|=\frac12$.
We will denote the allocation of $\M(F^{(1)})$ by $(X_1,X_2)$.
$X_1$ and $X_2$ will be used multiple times in the definitions of other instances.

\begin{definition}
$X_1=\M_1(F^{(1)})$ and $X_2=\M_2(F^{(1)})$.
\end{definition}

We have shown that $|X_1|=|X_2|=\frac12$.
It is helpful to assume $X_1=[0,0.5]$ and $X_2=(0.5,1]$ without loss of generality.

In the instances constructed later, we let $\varepsilon>0$ be a sufficiently small real number.

Next, we consider the following instance.
\begin{instance}\label{ins2}
$F^{(2)}=(f_1^{(2)},f_2^{(2)})$, where $f_1^{(2)}(x)=1$ for $x\in[0,1]$ and
$$f_2^{(2)}(x)=\left\{\begin{array}{ll}
    \varepsilon & x\in X_1 \\
    1 & x\in X_2
\end{array}\right..$$
\end{instance}

The following proposition shows that the only possible allocation output by $\M$ for Instance~\ref{ins2} is $(X_1,X_2)$.
\begin{proposition}\label{prop:ins2}
$\M(F^{(2)})=(X_1,X_2)$.
\end{proposition}
\begin{proof}
Firstly, we must have $|\M_2(F^{(2)})|\leq\frac12$.
Otherwise, agent $1$ will receive a subset of length strictly less than $1/2$.
Since agent $1$'s valuation is uniform on $[0,1]$, $\M$ is not proportional.

Secondly, we must have $X_2\subseteq\M_2(F^{(2)})$.
Suppose agent $2$ does not receive all of $X_2$, i.e., $|X_2\cap\M_2(F^{(2)})|<\frac12$.
Given that $|\M_2(F^{(2)})|\leq\frac12$, we have 
$$v_2\left(\M_2(F^{(2)})\right)=v_2\left(X_1\cap\M_2(F^{(2)})\right)+v_2\left(X_2\cap\M_2(F^{(2)})\right)$$
$$\qquad\leq\varepsilon\cdot\left(\frac12-|X_2\cap\M_2(F^{(2)})|\right)+1\cdot |X_2\cap\M_2(F^{(2)})|<\frac12.$$
On the other hand, if agent $2$ misreports his/her value density function to $f_2^{(1)}$ (instead of his/her true value density function $f_2^{(2)}$), the mechanism receives input $(f_1^{(2)},f_2^{(1)})$, which becomes Instance~\ref{ins1} since $f_1^{(1)}=f_1^{(2)}$.
In this case the allocation output is $(X_1,X_2)$, and agent $2$'s total value, in terms of his true valuation $f_2^{(2)}$, is $\frac12$.
Therefore, agent $2$ can receive more value by misreporting his/her value density function, and $\M$ cannot be truthful.

Putting these observations together, we have $X_2\subseteq\M_2(F^{(2)})$ and $|\M_2(F^{(2)})|\leq\frac12$, which implies $\M_2(F^{(2)})=X_2$.
Agent $1$ will then receive the remaining part of the cake which is just enough to guarantee proportionality: $\M_1(F^{(2)})=X_1$.
\end{proof}

The next instance we consider is slightly more complicated.
\begin{instance}\label{ins3}
$F^{(3)}=(f_1^{(3)},f_2^{(3)})$, where
$$f_1^{(3)}(x)=\left\{\begin{array}{ll}
    0.5 & x\in X_1 \\
    1 & x\in X_2
\end{array}\right.\qquad\mbox{and}\qquad
f_2^{(3)}(x)=\left\{\begin{array}{ll}
    \varepsilon & x\in X_1 \\
    1 & x\in X_2
\end{array}\right..$$
\end{instance}

The following proposition shows that each agent's allocated subset is exactly the union of half of $X_1$ and half of $X_2$.
\begin{proposition}\label{prop:ins3}
$|\M_1(F^{(3)})\cap X_1|=|\M_1(F^{(3)})\cap X_2|=|\M_2(F^{(3)})\cap X_1|=|\M_2(F^{(3)})\cap X_2|=\frac14$.
\end{proposition}
We provide a brief intuition behind the proof first.
Firstly, agent $1$ cannot receive a subset of length more than $0.5$.
Otherwise, in Instance~\ref{ins2}, agent $1$ will misreport his value density function from $f_1^{(2)}$ to $f_1^{(3)}$, which is more beneficial to agent $1$ (as $f_1^{(2)}$ is uniform and agent $1$ receives a larger length by misreporting).

Secondly, agent $1$ cannot receive less than half of $X_2$. If agent $1$ receives less than half of $X_2$ by a length of $x$, agent $1$ needs to receive more than half of $X_1$ by a length of at least $2x$ to guarantee proportionality.
This will make the total length received by agent $1$ more than $0.5$.

Thirdly, agent $1$ cannot receive more than half of $X_2$. Otherwise, we consider two cases. 
If agent $1$ receives a length of at least $3/8$ on $X_2$ (so that proportionality is already guaranteed for agent $1$), it is easy to see that proportionality cannot be guaranteed for agent $2$.
If agent $1$ receives a length between $1/4$ and $3/8$ on $X_2$, agent $2$, having significantly less value on $X_1$, will have to receive a length on $X_1$ that is significantly longer than half of $X_1$.
This will destroy the proportionality of agent $1$ for that agent $2$ has already taken too much.

Finally, having shown that agent $1$ must receive exactly half of $X_2$, the proportionality of agent $1$ and the proven fact that agent $1$'s received total length is at most $0.5$ imply that agent $1$ has to receive exactly half of $X_1$.
\begin{proof}[Proof of Proposition~\ref{prop:ins3}]
Firstly, we must have $|\M_1(F^{(3)})|\leq\frac12$.
Suppose this is not the case: $|\M_1(F^{(3)})|>\frac12$.
We show that $\M$ cannot be truthful.
Consider Instance~\ref{ins2} where agent $1$'s value density function is uniform.
In Instance~\ref{ins2}, if agent $1$ misreports his/her value density function to $f_1^{(3)}$, the mechanism $\M$ will see an input that is exactly the same as $F^{(3)}$ (notice $f_2^{(2)}=f_2^{(3)}$), and agent $1$ will receive a subset with length strictly more than $\frac12$.
However, we have seen in Proposition~\ref{prop:ins2} that agent $1$ will receive a subset with length exactly $\frac12$ if (s)he reports truthfully.
Since agent $1$'s true valuation is uniform, agent $1$ will benefit from this misreporting.

Let $|\M_1(F^{(3)})\cap X_2|=\frac14+x$ where $x\in[-\frac14,\frac14]$.
We aim to show that $x=0$.
Agent $1$'s total utility on $[0,1]$ is $\int_0^1f_1^{(3)}(x)dx=\frac34$.
To guarantee proportionality, we must have 
$$v_1\left(\M_1(F^{(3)})\right)=v_1\left(\M_1(F^{(3)})\cap X_1\right)+v_1\left(\M_1(F^{(3)})\cap X_2\right)$$
\begin{equation}\label{eqn:prop:ins3}
    \qquad=0.5\cdot\left|\M_1(F^{(3)})\cap X_1\right|+1\cdot\left(\frac14+x\right)\geq\frac38.
\end{equation}
By rearranging (\ref{eqn:prop:ins3}), we have $|\M_1(F^{(3)})\cap X_1|\geq\frac14-2x$.
The total length agent $1$ receives is then $|\M_1(F^{(3)})|=|\M_1(F^{(3)})\cap X_1|+|\M_1(F^{(3)})\cap X_2|\geq \frac12-x$.
Since we have seen $|\M_1(F^{(3)})|\leq\frac12$ at the beginning, we have $x\geq0$.

On the other hand, since $|\M_1(F^{(3)})\cap X_2|=\frac14+x$, we have $|\M_2(F^{(3)})\cap X_2|\leq\frac14-x$.
Since $v_2([0,1])=\frac12+\frac12\varepsilon$ and $v_2(\M_2(F^{(3)})\cap X_2)=1\cdot |\M_2(F^{(3)})\cap X_2|\leq\frac14-x$, to guarantee proportionality for agent $2$, we must have $v_2(\M_2(F^{(3)})\cap X_1)\geq\frac14\varepsilon+x$.
Therefore, $|\M_2(F^{(3)})\cap X_1|\geq\frac14+\frac{x}{\varepsilon}$, which implies $|\M_1(F^{(3)})\cap X_1|\leq\frac14-\frac{x}{\varepsilon}$.
Substituting this into (\ref{eqn:prop:ins3}), we have
$$0.5\cdot\left(\frac14-\frac{x}{\varepsilon}\right)+\left(\frac14+x\right)\geq\frac38,$$
which implies $x\leq0$ if $\varepsilon$ is sufficiently small.

Therefore, $x=0$, and we have $|\M_1(F^{(3)})\cap X_2|=\frac14$.
Since agent $1$ receives exactly length $\frac14$ on $X_2$, to guarantee proportionality, agent $1$ must receive at least length $\frac14$ on $X_1$.
To guarantee $|\M_1(F^{(3)})|\leq\frac12$, agent $1$ must receive at most length $\frac14$ on $X_1$.
Therefore, we have $|\M_1(F^{(3)})\cap X_1|=\frac14$.

Finally, agent $2$ must receive the remaining part of the cake to guarantee proportionality.
\end{proof}

We will define four subsets $X_{11},X_{12},X_{21},X_{22}$ of $[0,1]$ that will be used for constructing other instances later.
\begin{definition}
$X_{11}=|\M_1(F^{(3)})\cap X_1|$, $X_{12}=|\M_2(F^{(3)})\cap X_1|$, $X_{21}=|\M_1(F^{(3)})\cap X_2|$ and $X_{22}=|\M_2(F^{(3)})\cap X_2|$.
\end{definition}
Proposition~\ref{prop:ins3} implies $|X_{11}|=|X_{12}|=|X_{21}|=|X_{22}|=\frac14$. 
It is helpful for the readers to assume $X_{11}=[0,0.25]$, $X_{12}=(0.25,0.5]$, $X_{21}=(0.5,0.75]$ and $X_{22}=(0.75,1]$.

\begin{instance}\label{ins4}
$F^{(4)}=(f_1^{(4)},f_2^{(4)})$, where
$$f_1^{(4)}(x)=\left\{\begin{array}{ll}
    1 & x\in X_{11}\\
    \varepsilon & x\in X_{12}\\
    2\varepsilon & x\in X_{21}\\
    \varepsilon & x\in X_{22}
\end{array}\right.\qquad\mbox{and}\qquad
f_2^{(4)}(x)=\left\{\begin{array}{ll}
    \varepsilon & x\in X_1 \\
    1 & x\in X_2
\end{array}\right..$$
\end{instance}

We will show that $\M(F^{(3)})$ and $\M(F^{(4)})$ output the same allocation.
\begin{proposition}\label{prop:ins4}
$\M_1(F^{(4)})=X_{11}\cup X_{21}$ and $\M_2(F^{(4)})=X_{12}\cup X_{22}$.
\end{proposition}
\begin{proof}
Noticing that $f_2^{(2)}=f_2^{(3)}=f_2^{(4)}$, for the same reason in the proof of Proposition~\ref{prop:ins3}, we must have $|\M_1(F^{(4)})|\leq\frac12$.
Otherwise, agent $1$ in Instance~\ref{ins2} will misreport his/her true value density function $f_1^{(2)}$ to $f_1^{(4)}$.

On the other hand, if agent $1$ misreports his/her true value density function $f_1^{(4)}$ to $f_1^{(3)}$, the mechanism $\M$ will see the same input as $F^{(3)}$ and allocate $X_{11}\cup X_{21}$ to agent $1$.
With respect to agent $1$'s true valuation $f_1^{(4)}$, this is worth $\frac14+\frac\varepsilon2$.
To guarantee truthfulness, agent $1$ must receive a value of at least $\frac14+\frac\varepsilon2$ on $\M_1(F^{(4)})$: $v_1(\M_1(F^{(4)}))\geq \frac14+\frac\varepsilon2$.

Given that agent $1$ can receive a subset of length at most $\frac12$, the maximum value agent $1$ can receive is $\frac14+\frac\varepsilon2$, by receiving the two subsets $X_{11}$ and $X_{21}$ that are most valuable to agent $1$.
Therefore, $|\M_1(F^{(4)})|\leq\frac12$ and $v_1(\M_1(F^{(4)}))\geq \frac14+\frac\varepsilon2$ imply $\M_1(F^{(4)})=X_{11}\cup X_{21}$.

Finally, to guarantee proportionality, agent $2$ must receive the remaining part of the cake.
\end{proof}

\begin{instance}\label{ins5}
$F^{(5)}=(f_1^{(5)},f_2^{(5)})$, where $f_1^{(5)}(x)=1$ for $x\in[0,1]$ and 
$$
f_2^{(5)}(x)=\left\{\begin{array}{ll}
    1-\varepsilon & x\in X_{11} \\
    \varepsilon & x\in X_{12}\\
    1 & x\in X_{2}
\end{array}\right..$$
\end{instance}

We show that there is only possible output for $\M(F^{(5)})$ that guarantee both truthfulness and proportionality, with $\M(F^{(5)})=\M(F^{(1)})=\M(F^{(2)})$.

\begin{proposition}\label{prop:ins5}
$\M_1(F^{(5)})=X_1$ and $\M_2(F^{(5)})=X_2$.
\end{proposition}
\begin{proof}
Firstly, we must have $|\M_1(F^{(5)})|\geq\frac12$ to guarantee proportionality for agent $1$.
Therefore, $|\M_2(F^{(5)})|\leq\frac12$.
Secondly, if agent $2$ misreports his/her value density function to $f_2^{(2)}$, the mechanism $\M$ will see an input exactly the same as $F^{(2)}$, and will allocate $X_2$ to agent $2$.
This is worth $\frac12$ with respect to agent $2$'s true valuation $f_2^{(5)}$.
Therefore, we must have $v_2(\M_2(F^{(5)}))\geq\frac12$, for otherwise agent $2$ will misreport his/her value density function to $f_2^{(2)}$.
Given that agent $2$ can receive a length of at most $\frac12$, the maximum value (s)he can receive is $\frac12$, by receiving $X_2$ that is most valuable to agent $2$.
Therefore, $\M_2(F_5)=X_2$.
To guarantee proportionality for agent $1$, we must also have $\M_1(F^{(5)})=X_1$.
\end{proof}

Notice that, although we do not require entire allocations, the proportionality and truthfulness constraints make the output allocations of $\M$ for the first five instances entire.

Finally, we will consider our last instance below, and show that $\M$ cannot be both truthful and proportional for any allocation it outputs.

\begin{instance}\label{ins6}
$F^{(6)}=(f_1^{(6)},f_2^{(6)})$, where
$$f_1^{(6)}(x)=\left\{\begin{array}{ll}
    1 & x\in X_{11}\\
    \varepsilon & x\in X_{12}\\
    2\varepsilon & x\in X_{21}\\
    \varepsilon & x\in X_{22}
\end{array}\right.\qquad\mbox{and}\qquad
f_2^{(6)}(x)=\left\{\begin{array}{ll}
    1-\varepsilon & x\in X_{11} \\
    \varepsilon & x\in X_{12}\\
    1 & x\in X_{2}
\end{array}\right..$$
\end{instance}

We will analyze this instance in the following sub-section.

\subsection{Analysis of $\M(F^{(6)})$}
\label{sect:submain}
We show that $\M$ cannot output an allocation for Instance~\ref{ins6} that guarantees both truthfulness and proportionality.
This will give us a contradiction and proves Theorem~\ref{thm:main2}.
To show this, we begin by proving three propositions, and then show that they cannot be simultaneously satisfied.

\begin{proposition}\label{prop:i}
$|\M_2(F^{(6)})\cap X_2|\leq\frac14+\frac14\varepsilon$.
\end{proposition}
\begin{proof}
Suppose this is not the case: $|\M_2(F^{(6)})\cap X_2|>\frac14+\frac14\varepsilon$.
Consider Instance~\ref{ins4}.
By Proposition~\ref{prop:ins4}, we have $\M_2(F^{(4)})= X_{12}\cup X_{22}$, and agent $2$ can receive value $\frac14+\frac14\varepsilon$ (with respect to $f_2^{(4)}$).
By misreporting from $f_2^{(4)}$ to $f_2^{(6)}$, the mechanism $\M$ will see input $F^{(6)}$ and allocate $\M_2(F^{(6)})$ to agent $2$ with $|\M_2(F^{(6)})\cap X_2|>\frac14+\frac14\varepsilon$.
With respect to agent $2$'s true value density function $f_2^{(4)}$ in Instance~\ref{ins4}, this is worth more than $\frac14+\frac14\varepsilon$.
Therefore, $\M$ cannot be truthful.
\end{proof}

\begin{proposition}\label{prop:ii}
$v_1(\M_1(F^{(6)}))\geq\frac14+\frac14\varepsilon$ with respect to $f_1^{(6)}$.
\end{proposition}
\begin{proof}
Suppose agent $1$ misreports his/her true value density function $f_1^{(6)}$ to $f_1^{(5)}$.
The mechanism $\M$ will see input $F^{(5)}$, which will allocate $X_1$ to agent $1$ by Proposition~\ref{prop:ins5}.
This is worth $\frac14+\frac14\varepsilon$ to agent $1$.
Therefore, to guarantee truthfulness, we must have $v_1(\M_1(F^{(6)}))\geq\frac14+\frac14\varepsilon$.
\end{proof}

\begin{proposition}\label{prop:iii}
$v_2(\M_2(F^{(6)}))\geq\frac38$ with respect to $f_2^{(6)}$.
\end{proposition}
\begin{proof}
We have $v_2([0,1])=\frac14((1-\varepsilon)+\varepsilon)+\frac12\times1=\frac{3}4$.
The proposition follows by the proportionality of agent $2$.
\end{proof}

We first give an intuitive argument to show that Proposition~\ref{prop:i}, \ref{prop:ii} and \ref{prop:iii} cannot be all satisfied.
In $F^{(6)}$, agent $2$ has a value equal to or approximately equal to $1$ on each of the three segments $X_{11},X_{21}$ and $X_{22}$ and has a negligible value on $X_{12}$.
Proposition~\ref{prop:i} indicates that (s)he can receive at most (a little bit more than) half of $X_{21}\cup X_{22}$.
To guarantee proportionality (indicated by Proposition~\ref{prop:iii}), (s)he must receive approximately half of $X_{11}$.
On the other hand, by our construction of $f_1^{(6)}$, it is easy to see that Proposition~\ref{prop:ii} indicates that almost the entire $X_{11}$ needs to be given to agent $1$.
This gives a contradiction.

Formally, Proposition~\ref{prop:i} implies $v_2(\M_2(F^{(6)})\cap X_2)\leq\frac14+\frac14\varepsilon$.
Proposition~\ref{prop:iii} then indicates $v_2(\M_2(F^{(6)})\cap X_1)\geq\frac18-\frac14\varepsilon$.
Even if the entire $X_{12}$ is allocated to agent $2$ (which is worth $\frac14\varepsilon$), we still have
$$
\left|\M_2(F^{(6)})\cap X_{11}\right|\geq\frac{\frac18-\frac14\varepsilon-\frac14\varepsilon}{1-\varepsilon}=\frac{1-4\varepsilon}{8-8\varepsilon}.
$$
For agent $1$, we must then have
$$\left|\M_1(F^{(6)})\cap X_{11}\right|\leq\frac14-\frac{1-4\varepsilon}{8-8\varepsilon}=\frac{1+2\varepsilon}{8-8\varepsilon}.$$
To find an upper bound for $v_1(\M_1(F^{(6)}))$, suppose agent $1$ receives all of $X_{12},X_{21}$ and $X_{22}$.
Even in this case, we have the following upper bound for $v_1(\M_1(F^{(6)}))$:
$$v_1(\M_1(F^{(6)}))\leq \frac{1+2\varepsilon}{8-8\varepsilon}\cdot 1 + \frac14\cdot\varepsilon +\frac14\cdot2\varepsilon+ \frac14\cdot\varepsilon=\frac{1+2\varepsilon}{8-8\varepsilon}+\varepsilon.$$
Taking $\varepsilon\rightarrow0$, the limit of the above upper bound is $\frac18$.
Thus, $v_1(\M_1(F^{(6)}))<\frac14+\frac14\varepsilon$ for sufficiently small $\varepsilon$, and Proposition~\ref{prop:ii} cannot be satisfied.

This concludes the proof of Theorem~\ref{thm:main2}.

\subsection{Truthful, Approximately Proportional Mechanisms}
We have just proved that a truthful, proportional mechanism does not exist.
To deploy cake cutting mechanisms in practice, it leaves us to consider relaxations on truthfulness or proportionality.

In the theorem below, we show the non-existence of approximately proportional mechanisms if we do not relax the dominant strategy truthfulness.
In the next three sections, we will consider some relaxations on truthfulness and provide some mechanisms satisfying the relaxed truthfulness (while guaranteeing fairness).

\begin{theorem}\label{thm:main2_approx}
There does not exist a truthful and $0.974031$-approximately proportional mechanism, even when all of the followings hold:
\begin{itemize}
    \item there are two agents;
    \item each agent's value density function is piecewise-constant;
    \item each agent is hungry: each $f_i$ satisfies $f_i(x)>0$ for any $x\in[0,1]$;
    \item the mechanism needs not to be entire.
\end{itemize}
\end{theorem}

The proof of the above theorem is similar to the proof of Theorem~\ref{thm:main2}, with the addition of many approximation analyses.
We defer it to \ref{append:proof_approx}.

The existence of truthful and approximately proportional mechanism with smaller approximation ratios is still an open problem, and we will discuss more about it in Sect.~\ref{sect:conclusion}.

\section{Monotone Value Density Functions}
\label{sect:monotone}
In this section, we demonstrate that the impossibility result in Theorem~\ref{thm:main2} fails when agents' value density functions are piecewise-constant and \emph{monotone}.
Monotone valuations are natural in many applications where agents' interest on the resource is decreasing or increasing due to some special properties of the resource.
For example, when allocating the advertisement slots on a web page, the slots on the top of the page always have larger values to the agents, although some agents may value those top slots higher than the others.
Under this setting, the mechanism takes $n$ piecewise-constant and monotone value density functions as inputs.
In particular, an agent is forbidden to report a non-monotone value density function (we can assume the mechanism will always first convert a non-monotone piecewise-constant function to a monotone one based on a consistent rule).

\subsection{A Truthful Envy-Free Mechanism for Monotone Value Density Functions}
In this section, we present a truthful envy-free mechanism for monotone value density functions.

\begin{theorem}\label{thm:monotone}
If all the agents' value density functions are piecewise-constant and increasing (or decreasing), there exists a mechanism that is truthful, entire, and envy-free (and thus proportional). 
\end{theorem}

Without loss of generality, we can suppose all the value density functions are increasing. 
Our mechanism is presented in Mechanism~\ref{alg:simple_mono}.
The mechanism consists of $n$ iterations, and the allocation for agent $i$ is determined at the $i$-th iteration.
At the first iteration, agent $1$ cuts the cake at the points of discontinuity of $f_1$, so that the cake is split into multiple intervals and $f_1$ is uniform on each interval.
After that, the leftmost $1/n$ fraction of each interval is allocated to agent $1$.
At the $i$-th iteration, the unallocated part of the cake may consist of multiple intervals, and these intervals are further sub-divided into more intervals where $f_i$ is uniform on each of them.
On each of these intervals, agent $i$ receives the leftmost $1/|S|$ fraction, where $S$ is the set of the agents who have not been allocated (in particular, $|S|=n-i+1$).

\begin{algorithm}
\caption{An envy-free and truthful cake cutting algorithm for increasing functions}
\label{alg:simple_mono}
\begin{algorithmic}[1]
\STATE initialize $U=\{[0,1)\}$\; // $U$ is a collection of intervals that are currently unallocated
\STATE initialize $S=[n]$\; // $S$ is the set of agents who have not been allocated
\STATE \textbf{for} each $i=1,\ldots,n$:
\STATE \hspace{0.5cm} initialize agent $i$'s allocation $A_i\leftarrow\emptyset$\;
\STATE \hspace{0.5cm} \textbf{for} each interval $[s,t)\in U$, 
\STATE \hspace{1cm} let $x_1,x_{2},\ldots,x_{k}$ be the discontinuity points of $f_i$ that are in the interval $[s,t)$\;
\STATE \hspace{1cm} split $[s,t)$ to $k+1$ intervals $I_0=[s,x_1),I_1=[x_1,x_2),\ldots,I_k=[x_k,t)$\;
\STATE \hspace{1cm} for each $I_j=[y,z)$ with $j=0,1,\ldots,k$, update agent $i$'s allocation by $A_i\leftarrow A_i\cup[y, y+\frac{z-y}{|S|})$\;
\STATE \hspace{0.5cm} \textbf{endfor}
\STATE \hspace{0.5cm} update $S\leftarrow S\setminus\{i\}$\;
\STATE \hspace{0.5cm} update $U$ by replacing each $I\in U$ with the interval(s) in $I\setminus A_i$ (remove $I$ from $U$ if $I\subseteq A_i$)\;
\STATE \textbf{endfor}
\STATE \textbf{return} the allocation $(A_1,\ldots,A_n)$\;
\end{algorithmic}
\end{algorithm}

To simplify later analysis, we use $U_i$ and $S_i$ to represent the value of $U$ and $S$ before the $i$-th iteration, respectively (in particular, $U_1 = \{[0, 1)\}$ and $S_1 = [n]$). 
Given a collection $\mathcal{I}$ of intervals and a value density function $i$, we slightly abuse the notation and let
$$v_i(\mathcal{I})=\sum_{I\in\mathcal{I}}v_i(I).$$

We first show the following proposition.
\begin{proposition}\label{prop:simple_mono_lemma1}
For each iteration $i=1,\ldots,n$, the followings are true.
\begin{enumerate}
    \item $|A_i|=\frac1n$
    \item $v_j(A_i)\leq\frac1{|S_i|}v_j(U_i)$ for any $j=1,\ldots,n$.  
    \item $v_j(A_i)=\frac1{|S_j|}v_j(U_j)$ for any $j\leq i$.
\end{enumerate}
\end{proposition}
\begin{proof}
Firstly, we argue that $\sum_{u\in U_i} |u| = \frac{n+1-i}{n}$ for any $i\in[n]$. We can prove it by induction. When $i=1$, $\sum_{u\in U_1} |u| = 1$. Assume the inductive hypothesis is true when $i=\ell$. When $i=\ell+1$, since the agent $(i-1)$ receives intervals with total length of exactly $\frac{1}{n-(i-2)}\cdot\left(\sum_{u\in U_{i-1}} |u|\right)$, $\sum_{u\in U_i} |u|$ equals to $\left(\sum_{u\in U_{i-1}} |u|\right) \cdot \left(1 - \frac{1}{n-(i-2)}\right) = \frac{n+1-i}{n}$. Furthermore, it can also immediately imply that each agent will receive intervals with the same total length $\frac1n$. This is because agent $i$ receives exactly $\frac{1}{n+1-i}$ of $U_i$ for any $1\le i \le n$.
We conclude 1.

During the $i$-th iteration, on each interval $[y,z)\in U_i$, agent $i$ receives $[y, y+\frac{z-y}{|S_i|})$ which is a $\frac1{|S_i|}$ fraction of the length.
In addition, (s)he only receives the leftmost part of each interval.
We have $v_j\left([y, y+\frac{z-y}{|S_i|})\right)\leq\frac1{|S_i|}v_j([y,z))$ as $v_j$ is increasing.
Summing up all the intervals in $U_i$, we have $v_j(A_i) \le \frac{1}{|S_i|}v_j(U_i)$ for any agent $j$, which concludes 2.

Now, consider the $j$-th iteration.
For each $[y,z)\in U_j$, agent $j$ receives $[y, y+\frac{z-y}{|S_j|})$.
In addition, by our mechanism, $v_j$ is uniform on each $[y,z)\in U_j$.
Therefore, $v_j(A_j)=\frac1{|S_j|}v_j(U_j)$.
This proves 3 for $j=i$.

For each $i>j$, by using a similar inductive argument as it is in the first paragraph, it can be easily proved that agent $i$ receives a $\frac1{|S_j|}$ fraction of length on $[y,z)$.
(In particular, $[y,z)$ may be further divided into multiple intervals at later iterations, but it is always the case that an agent at a later iteration will receive the average length on each of these intervals.)
Since $v_j$ is uniform on $[y,z)$, the interval that agent $i$ receives on $[y,z)$ is worth exactly $\frac1{|S_j|}v_j\left([y,z)\right)$ in terms of agent $j$'s valuation.
Summing up all the intervals in $U_j$, we have $v_j(A_i)=\frac1{|S_j|}v_j(U_j)$, which concludes 3.
\end{proof}

We next show that the mechanism is entire and envy-free.
\begin{lemma}
Mechanism~\ref{alg:simple_mono} is entire and envy-free (and thus proportional).
\end{lemma}
\begin{proof}
First, in the last iteration $i=n$, we have $|S|=1$, and it is easy to see that agent $n$ receives the remaining part of the cake.
Thus, the mechanism is entire.

Next, we prove Mechanism~\ref{alg:simple_mono} is envy-free. 
The third part of Proposition~\ref{prop:simple_mono_lemma1} immediately implies $v_j(A_j) = v_j(A_{j+1}) = \cdots = v_j(A_n)$, so agent $j$ does not envy any of $j+1,\ldots,n$.

Hence, we just need to demonstrate that $v_j(A_j)$ is the maximum among $v_j(A_1), \ldots, v_j(A_j)$. According to the second part of Proposition~\ref{prop:simple_mono_lemma1}, we have $v_j(A_i) \le \frac{1}{|S_i|}v_j(U_i) = \frac{1}{|S_i|}\left(v_j(A_i) + v_j(U_{i+1})\right)$, which implies 
\begin{equation}\label{eqn:monotoneEF}
    v_j(A_i) \le \frac{1}{|S_{i+1}|} v_j(U_{i+1}).
\end{equation}
For agent $j-1$, we have
$$
v_j(A_{j-1}) \le \frac{1}{n+1-j} v_j(U_j) = v_j(A_j),
$$
where the last equality is due to 3 of Proposition~\ref{prop:simple_mono_lemma1}.
This shows agent $j$ does not envy agent $j-1$.

For agent $j-2$, we have
\begin{align*}
v_j(A_{j-2}) &\le \frac{1}{n+2-j} v_j(U_{j-1}) \tag{by (\ref{eqn:monotoneEF})}\\
&= \frac{1}{n+2-j}\left( v_j(A_{j-1}) + v_j(U_{j}) \right)\\
&\leq\frac{1}{n+2-j}\left( v_j(A_{j}) + v_j(U_{j}) \right)\tag{we have shown agent $j$ does not envy agent $j-1$}\\
&=\frac{1}{n+2-j}\left( v_j(A_{j}) + (n+1-j)v_j(A_{j}) \right)\tag{3 of Proposition~\ref{prop:simple_mono_lemma1}}\\
&=v_j(A_j),
\end{align*}
which implies agent $j$ does not envy agent $j-2$.

This analysis can be continued, and at last for agent $1$, we have
\begin{align*}
    v_j(A_1) &\le \frac{1}{n-1} v_j(U_{2})\tag{by (\ref{eqn:monotoneEF})}\\
    &=\frac{1}{n-1}\left(v_j(A_2) + \cdots + v_j(A_{j-1}) + v_j(U_{j}) \right)\\
    &\leq\frac1{n-1}\left((j-2)v_j(A_j)+v_j(U_j)\right)\tag{we have shown agent $j$ does not envy any of $2,\ldots,j-1$}\\
    &=\frac1{n-1}\left((j-2)v_j(A_j)+(n+1-j)v_j(A_{j}) \right)\tag{3 of Proposition~\ref{prop:simple_mono_lemma1}}\\
    &=v_j(A_j).
\end{align*}

Therefore, we can conclude that agent $i$ does not envy other agents, which means Mechanism~\ref{alg:simple_mono} is envy-free.
Since the mechanism is envy-free and entire, it is proportional.
\end{proof}
Finally, we show that this mechanism is truthful.
\begin{lemma}
Mechanism~\ref{alg:simple_mono} is truthful.
\end{lemma}
\begin{proof}
According to Mechanism~\ref{alg:simple_mono}, the intervals allocated to agent $1,\ldots,i-2$ and $i-1$ will not change no matter how agent $i$ misreports his(her) valuation. On the other hand, we can find that the intervals allocated to agent $i$ only depend on the discontinuity points of his(her) reported value density function. For this reason, we will only discuss the cases that agent $i$ misreports his(her) discontinuity points.

Assume $U_i = \{ [s_1, t_1), \ldots, [s_k, t_k)\}$. Suppose the discontinuity points of agent $i$'s reported value density function split $U_i$ into the set of intervals $\mathcal{I}$. Since agent $i$ will receive the leftmost part of each interval in $\mathcal{I}$ and $v_i$ is monotone, the total value agent $i$ can receive is at most:
$$
\sum_{I\in \mathcal{I}} \frac{v_i(I)}{|S_i|}  = \frac{1}{|S_i|}v_i\left(\mathcal{I}\right)  = \frac{1}{|S_i|}v_i\left(U_i\right).
$$

However, according to the third part of Proposition~\ref{prop:simple_mono_lemma1}, agent $i$ can already receive a set of intervals with value $\frac{1}{|S_i|}v_i\left(U_i\right)$, so there is no need for agent $i$ to misreport his(her) valuation.
Therefore, Mechanism~\ref{alg:simple_mono} is truthful.
\end{proof}

The two lemmas above conclude Theorem~\ref{thm:monotone}.

We remark that the success of Mechanism~\ref{alg:simple_mono} depends crucially on its \emph{non-anonymity}.
A mechanism is \emph{anonymous} if agent's received value does not depend on agent's index.
Formally, a mechanism $\M$ is anonymous if, given any permutation $\pi:[n]\to[n]$ and letting $(A_1,\ldots,A_n)$ and $(B_1,\ldots,B_n)$ be the outputs of $\M(f_1,\ldots,f_n)$ and $\M(f_{\pi(1)},\ldots,f_{\pi(n)})$ respectively, we have $f_i(A_i)=f_i(B_{\pi^{-1}(i)})$ for each $i=1,\ldots,n$.
Notice that Mechanism~\ref{alg:simple_mono} is highly non-anonymous, as agents' indices play a crucial role on their allocations.
Although anonymity appears to be a mild assumption that is held by many natural mechanisms, it is incompatible with truthfulness even if we completely disregard fairness and even for some very special cases.
Bei et al.~\cite{bei2020truthful} show that there does not exist a truthful and anonymous mechanism even if each agent's value density function $f_i$ is in the following form that is characterized by only one parameter $s_i$
$$f_i(x)=\left\{\begin{array}{ll}
    1 & x\in[0,s_i) \\
    0 & x\in[s_i,1]
\end{array}\right..$$
Notice that this type of functions is a special case of monotone functions.

\setcounter{instance}{0}

\subsection{Incompatibility with Pareto-Optimality}
In this section, we show that truthful and fairness is incompatible with Pareto-optimality, even for piecewise-constant and monotone value density functions.

\begin{theorem}\label{thm:monotonePO}
For any $\alpha>0$, there does not exist a mechanism that is truthful, $\alpha$-proportional, and Pareto-optimal, even when each agent's value density function is piecewise-constant and increasing (or decreasing), and contains at most one point of discontinuity.
\end{theorem}
\begin{proof}
We will prove Theorem~\ref{thm:monotonePO} by contradiction as before. 
Suppose there exists a truthful, $\alpha$-approximately proportional and Pareto-optimal mechanism $\M$.
We will construct three instances with two agents, analyze the outputs of $\M$ on these instances, and prove that truthfulness, proportionality and Pareto-optimality cannot be guaranteed on all these instances.
The three instances are shown in Table~\ref{tab_PO}.

We start with the same instance as Instance~\ref{ins1} in Sect.~\ref{sect:impossibility}.
\begin{instance}\label{ins7}
$F^{(1)}=(f_1^{(1)},f_2^{(1)})$, where $f_1^{(1)}(x)=1$ and $f_2^{(1)}(x)=1$ for $x\in[0,1]$.
\end{instance}

To guarantee $\alpha$-approximate proportionality, we must have $|\M_1(F^{(1)})| \geq \frac\alpha2$ and $|\M_2(F^{(1)})| \geq \frac{\alpha}{2}$. 
The interval $[0,1]$ will be separated into several sub-intervals $\{U_0, \dots, U_t, X\}$, where $X=[s, 1]$ for some $s\in(0, 1)$ is the rightmost interval, and each interval is allocated to exactly one agent.
Without loss of generality, we assume agent $2$ receives $X$.
We have $s\geq\frac\alpha2$ to ensure $|\M_1(F^{(1)})| \geq \frac\alpha2$. 

We use $Y$ for the interval $[0, s)$.

\begin{instance}\label{ins8}
$F^{(2)}=(f_1^{(2)},f_2^{(2)})$, where $f_1^{(2)}(x)=1$ for $x\in[0,1]$ and
$$f_2^{(2)}(x)=\left\{\begin{array}{ll}
    0 & x\in Y \\
    1 & x\in X
\end{array}\right..$$
\end{instance}

Firstly, agent $1$ must receive $Y$. 
If any part of $Y$ is given to agent $2$, it violates Pareto-optimality. 
Secondly, agent $2$ must receive $X$. 
Otherwise, agent $2$ will misreport his/her true value density function from $f_2^{(2)}$ to $f_2^{(1)}$ to obtain an allocation that contains $X$, which increases his/her value. 
Hence, $M_1(F^{(2)}) = Y$ and $M_2(F^{(2)}) = X$.

The allocation is also $\alpha$-approximately proportional. 
For agent $2$, (s)he receives the maximum value. 
For agent $1$, $|Y|\ge\frac{\alpha}{2}$, so $v_1(Y)\ge\frac{\alpha}{2}$.

\begin{instance}\label{ins9}
$F^{(3)}=(f_1^{(3)},f_2^{(3)})$, where
$$f_1^{(3)}(x)=\left\{\begin{array}{ll}
    \varepsilon & x\in Y \\
    1 & x\in X
\end{array}\right.\qquad\mbox{and}\qquad
f_2^{(3)}(x)=\left\{\begin{array}{ll}
    0 & x\in Y \\
    1 & x\in X
\end{array}\right..$$
Moreover, $\varepsilon$ is set to be small enough such that  $\varepsilon<\frac{(1-s)\alpha}{(2-\alpha)s}$ for the $\alpha$ given in the theorem.
\end{instance}

Again, agent $1$ must receive $Y$ to ensure Pareto-optimality. 
Moreover, since $\varepsilon$ is extremely small, it does not satisfy $\alpha$-proportionality for agent $1$ if (s)he only receives $Y$, so (s)he must receive additional part from $X$. 
In this case, however, when considering Instance~\ref{ins8}, by misreporting from $f_1^{(2)}$ to $f_1^{(3)}$, agent $1$ will receive higher value and the mechanism cannot be truthful.

Hence, we conclude from this instance that truthfulness, proportionality and Pareto-optimality cannot be guaranteed at the same time.
\end{proof}

\begin{table}
    \centering
    \begin{tabular}{|l|l|}
    \hline
    Instance & Allocation \\
    \hline
      \includegraphics{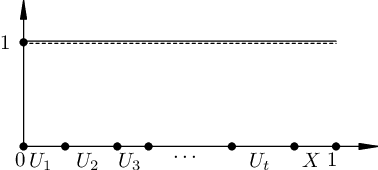}   &  
      \begin{tabular}{l}
           $\M_1(F^{(1)})$: some collections of $U_i$'s  \\
           $\M_2(F^{(1)})$: some collections of $U_i$'s and $X$ 
      \end{tabular}\\
      \hline
      \includegraphics{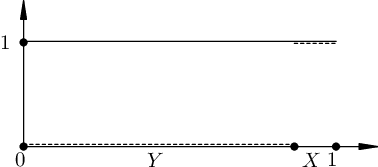}  & $\M_1(F^{(2)})=Y$ and $\M_2(F^{(2)})=X$\\
      \hline
      \includegraphics{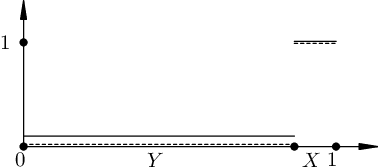}  & See the proof of Theorem~\ref{thm:monotonePO}\\
      \hline
    \end{tabular}
    \caption{Instances constructed for the proof of Theorem~\ref{thm:monotonePO} and the corresponding allocations given by $\M$. The value density for agent $1$ is shown in solid lines, and the value density for agent $2$ is shown in dashed lines.}
    \label{tab_PO}
\end{table}

As a remark, Theorem~\ref{thm:monotonePO} holds if Pareto-optimality is replaced by \emph{non-wastefulness}, with exactly the same proof.
An allocation is \emph{non-wasteful} if each $A_i$ does not contain any part where the value of $f_i$ is zero.
Notice that non-wastefulness is a weaker efficiency guarantee than Pareto-optimality.

\section{On Weaker Truthful Guarantees, Risk-Averse Truthfulness}
\label{sect:RAT}
We have seen in the previous section that standard dominant strategy truthfulness cannot be guaranteed if we want a proportional mechanism or even an approximately proportional mechanism with a sufficiently large approximation ratio.
In this section, we will consider weaker truthful criteria.

One natural idea of relaxing truthfulness is to consider \emph{approximation on truthfulness}, where an agent will not receive a utility that is more than $\alpha$ times the utility (s)he would have received by truth-telling.
However, such a notion is unconvincing in the game theory aspect, although it may be compatible in the spirit of approximation algorithm.
An agent will still misreport his/her valuation under an $\alpha$-approximately truthful mechanism.
On the other hand, there may be other much more stable equilibria than the truth-telling profile.
Agents' behaviors are still largely unpredictable under an $\alpha$-approximately truthful mechanism.
Therefore, we seek some other relaxation on truthfulness.

A common truthful criterion is to require that the truth-telling profile form a \emph{Nash Equilibrium}.
In many applications, this is a significantly weaker guarantee than dominant strategy truthfulness.
However, in our cake cutting case with direct revelation mechanisms, this truthful criterion is equivalent to the dominant strategy truthfulness, as the following theorem shows.

\begin{theorem}\label{thm:Nash}
If a mechanism $\M$ satisfies that agents' strategies of truthfully reporting their value density functions form a Nash equilibrium, then $\M$ is (dominant strategy) truthful.
\end{theorem}
\begin{proof}
Suppose $\M$ satisfying this property is not dominant strategy truthful.
Given a valuation profile $(f_1,\ldots,f_n)$, there must exist an agent $i$ and $n-1$ value density functions $f_1',\ldots,f_{i-1}',f_{i+1}',\ldots,f_n'$ reported by the other $n-1$ agents, such that reporting certain $f_i'$ is more beneficial for agent $i$ than truthfully reporting $f_i$.
Now, consider a different valuation profile $(f_1',\ldots,f_{i-1}',f_i,f_{i+1}',\ldots,f_n')$.
In this new profile, for each $j\neq i$, the function $f_j'$, being the reported function in the previous case, becomes the true valuation for agent $j$.
In this new setting, if the remaining $n-1$ agents truthfully report their value density functions, which are $f_1',\ldots,f_{i-1}',f_{i+1}',\ldots,f_n'$, agent $i$'s best response is to report $f_i'$ instead of his/her true valuation $f_i$ (as we have seen in the first setting).
This indicates that truth-telling is not a Nash equilibrium.
\end{proof}

We remark that the theorem above holds in a very general context: it holds in every normal form game where each strategy profile can represent a truthful profile (i.e., for each strategy profile $(s_1,\ldots,s_n)$, there is an instance where, for each $i=1,\ldots,n$, $s_i$ is the ``truthful strategy'' for agent $i$).

\medskip

Even though we do not have any progress on many standard truthful guarantees in game theory, there are still mechanisms that can achieve ``a certain degree of truthfulness'' in practice.
Most notably, the \emph{I-cut-you-choose} protocol achieves some kind of truthfulness.
The protocol works for proportional/envy-free cake cutting with two agents: agent $1$ find a point $x$ such that $v_1([0,x])=v_1([x,1])$; agent $2$ is allocated one of $[0,x]$ and $[x,1]$ that is more valuable to him/her, and the other piece is allocated to agent $1$.
It is easy to see that agent~$2$'s dominant strategy is truth-telling: (s)he has no control on the position of $x$, and truth-telling can ensure (s)he gets a piece with a larger value.
On the other hand, although it is not a dominant strategy for agent $1$ to tell the truth, agent $1$ still does not have the incentive to lie in the case (s)he has no knowledge of agent $2$'s valuation.
If (s)he reports a value density function that results in a different position of $x$, there is always a risk that (s)he will receive a piece with a value less than $1/2$ of the entire cake (i.e., less than the value guaranteed by proportionality).

There are two reasons behind agent $1$'s truth-telling incentive.
Firstly, as mentioned, (s)he does not have prior knowledge on agent $2$'s valuations.
Secondly, (s)he is a risk-averse agent: whenever there is a risk of receiving a value that is less than what (s)he would have received by truth-telling, (s)he prefers to avoid the risk.

Motivated by this example, we define and consider a new truthful criterion: \emph{the risk-averse truthfulness}.

\begin{definition}\label{def:wrat}
A mechanism $\M$ is \emph{risk-averse truthful} if, for each agent $i$ with value density function $f_i$ and for any $f_i'$, either one of the following holds:
\begin{enumerate}
    \item for any $f_1,\ldots,f_{i-1},f_{i+1},\ldots,f_n$,
    $$v_i(\M_i(f_1,\ldots,f_{i-1},f_i,f_{i+1},\ldots,f_n))\geq v_i(\M_i(f_1,\ldots,f_{i-1},f_i',f_{i+1},\ldots,f_n));$$
    \item there exist $f_1,\ldots,f_{i-1},f_{i+1},\ldots,f_n$ such that $$v_i(\M_i(f_1,\ldots,f_{i-1},f_i',f_{i+1},\ldots,f_n))<v_i(\M_i(f_1,\ldots,f_{i-1},f_i,f_{i+1},\ldots,f_n)).$$
\end{enumerate}
\end{definition}

In other words, a mechanism is risk-averse truthful if either an agent's misreporting is non-beneficial, or the misreporting can potentially cause the agent to receive a piece with a value that is less than what (s)he would have received by truth-telling.

The I-cut-you-choose protocol can achieve a stronger truthful property: if agent $1$ modifies the cut-point $x$ by misreporting his/her value density function, there is always a chance that (s)he will receive a piece with a value that is even less than his/her proportional value.
Motivated by this, we define a stronger truthful notion that based on the fairness criterion of proportionality.

\begin{definition}\label{def:rat}
A mechanism $\M$ is \emph{proportionally risk-averse truthful} if 
\begin{enumerate}
    \item $\M$ is proportional, and
    \item for each agent $i$ with value density function $f_i$ and for any $f_i'$, either one of the following holds:
    \begin{enumerate}
    \item for any $f_1,\ldots,f_{i-1},f_{i+1},\ldots,f_n$,
    $$v_i(\M_i(f_1,\ldots,f_{i-1},f_i,f_{i+1},\ldots,f_n))\geq v_i(\M_i(f_1,\ldots,f_{i-1},f_i',f_{i+1},\ldots,f_n));$$
    \item there exist $f_1,\ldots,f_{i-1},f_{i+1},\ldots,f_n$ such that $$v_i(\M_i(f_1,\ldots,f_{i-1},f_i',f_{i+1},\ldots,f_n))<\frac1n v_i([0,1]).$$
    \end{enumerate}
\end{enumerate}
\end{definition}

It is clear that proportional risk-averse truthfulness implies risk-averse truthfulness, as 1 and 2(b) in Definition~\ref{def:rat} imply 2 in Definition~\ref{def:wrat}.
On the other hand, any truthful mechanism without proportional guarantee (e.g., the mechanism allocating the entire cake to a single agent) is risk-averse truthful, but it is not proportionally risk-averse truthful by definition.
Even if we are restricted to proportional mechanisms, there are mechanisms that are risk-averse truthful but not proportionally risk-averse truthful (see Table~\ref{tab:implication} and Table~\ref{tab:mechanisms}).
Thus, proportional risk-averse truthfulness is a strictly stronger notion than risk-averse truthfulness.
In particular, a proportionally risk-averse truthful mechanism is robust against ``slightly risk-seeking'' agents who can accept potentially receiving less values, as long as the received value is at least proportional.



Brams, Jones, and Klamler~\cite{brams2006better} also define a truthful notion in the spirit of agents' risk-averseness and uncertainty about other agents' valuations, for which we will call it \emph{maximin strategy-proofness}\footnote{Brams, Jones, and Klamler call it \emph{strategy-proofness} in their paper~\cite{brams2006better}. However, strategy-proofness now commonly refers to dominant-strategy truthfulness (the truthfulness defined in Sect.~\ref{sect:prelim}). We use the name \emph{maximin strategy-proofness} following the papers~\cite{brams2006better,ortega2022obvious}.}.
Their notion is weaker than our risk-averse truthfulness (and so further weaker than the proportional risk-averse truthfulness).
In \ref{append:Brams}, we will discuss the difference between our truthful notions and theirs, and we will also point out a minor mistake made in their paper.

We remark that there are other truthful notions that relax the dominant-strategy truthfulness with the consideration of agents' uncertainty about each other's utility.
For example, Troyan and Morrill~\cite{troyan2020obvious} define a truthful notion called ``not obviously manipulatable'' which requires that manipulation should not be strictly better off in both the worst case and the best case.
Besides many technical differences, Troyan and Morrill's notion is also conceptually different from our (proportionally) risk-averse truthfulness.
The (proportionally) risk-averse truthfulness puts more focus on agents’ risk-averseness, whereas more focus is put on the difficulty of finding a deviation in Troyan and Morrill's notion.
We will formally define this notion in \ref{append:nom}.
Comparing the strength of our notion with Troyan and Morrill's, neither one implies the other.
We will also formally show this in \ref{append:nom}.

The relationships between these truthful notions are shown in Table~\ref{tab:implication}.

\begin{table}[ht]
    \centering
    \begin{tabular}{lc}
    \toprule
    {\bf Implication} & {\bf Justification}\\
    \hline
     PRAT $\Rightarrow$ RAT+Proportional  &  See the remark following Def~\ref{def:rat}\\
     RAT+Proportional $\not\Rightarrow$ PRAT    & e.g., Ortega and Segal-Halevi's Moving-Knife (Thm~\ref{thm:OS_PRAT} and Thm~\ref{thm:OS_RAT})\\
     RAT $\Rightarrow$ maximinSP & See the remark following Def~\ref{def:maximinSP}\\
     maximinSP $\not\Rightarrow$ RAT & e.g., Dubins-Spanier's Moving-Knife (Thm~\ref{thm:moving-knife} and Thm~\ref{thm:maximinSP_proportional})\\
     PRAT $\not\Rightarrow$ NOM & e.g., Dubins-Spanier's Moving-Knife (Thm~\ref{thm:moving-knife} and Thm~\ref{thm:moving_knife_nom})\\
     NOM $\not\Rightarrow$ RAT & e.g., Mechanism~\ref{alg:rat_ef} (Thm~\ref{thm:rat_ef} and Thm~\ref{thm:rat_ef_nom})\\
     \bottomrule
    \end{tabular}
    \caption{Implications of truthful notions. PRAT stands for proportional risk-averse truthfulness, RAT stands for risk-averse truthfulness, maximinSP stands for maximin strategy-proofness (\cite{brams2006better}), and NOM stands for not obvious manipulability (\cite{troyan2020obvious}).}
    \label{tab:implication}
\end{table}

In Table~\ref{tab:mechanisms}, we present the satisfabilities of the four above-mentioned truthful notions for many (deterministic) mechanisms discussed in this paper, including the two classical mechanisms, the Dubins-Spanier moving-knife procedure~\cite{dubins1961cut} and the Even-Paz algorithm~\cite{even1984note} (see sect.~\ref{sect:RAT_P} for their descriptions), the equitability procedure proposed by Brams, Jones and Klamler~\cite{brams2006better} (see \ref{append:Brams2} for its description), the variant of the moving-knife procedure proposed by Ortega and Segal-Halevi~\cite{ortega2022obvious} (see \ref{append:OSmovingknife} for its description), and the three mechanisms proposed in this paper (see Sect.~\ref{sect:RAT_EF} and Sect.~\ref{sect:RAT_P} for their descriptions).

\begin{table}[ht]
    \centering
    \begin{tabular}{lcccc}
    \toprule
    {\bf Mechanism} & PRAT & RAT & maximinSP & NOM \\
    \hline
    Dubins-Spanier's Moving-Knife~\cite{dubins1961cut} & $\times$ (Thm~\ref{thm:moving-knife}) & $\times$ (Thm~\ref{thm:moving-knife}) & $\checkmark$ (Thm~\ref{thm:maximinSP_proportional}) & $\checkmark$ (Thm~\ref{thm:moving_knife_nom}) \\
    Even-Paz~\cite{even1984note} & $\times$ (Thm~\ref{thm:even-paz}) & $\checkmark/\times^\dag$ (\ref{append:EvenPaz}) & $\checkmark$ (Thm~\ref{thm:maximinSP_proportional}) & $\checkmark$ (Thm~\ref{thm:Even-Paz_nom})  \\
    Equitability Procedure~\cite{brams2006better} & $\times$ (Thm~\ref{thm:equitability_PRAT}) & ? & $\checkmark$ \cite{brams2006better} & $\checkmark$ (Thm~\ref{thm:equitability_nom}) \\
    Ortega and Segal-Halevi's Moving-Knife~\cite{ortega2022obvious} & $\times$ (Thm~\ref{thm:OS_PRAT}) & $\checkmark/\times^\dag$ (\ref{append:OSRAT}) & $\checkmark$ (Thm~\ref{thm:maximinSP_proportional}) & $\checkmark$ \cite{ortega2022obvious} \\
    Mechanism~\ref{alg:rat_ef} & $\checkmark$ (Thm~\ref{thm:rat_ef}) & $\checkmark$ (Thm~\ref{thm:rat_ef}) & $\checkmark$ (Thm~\ref{thm:maximinSP_proportional}) & $\times$ (Thm~\ref{thm:rat_ef_nom}) \\
    Mechanism~\ref{alg:rat_p} (assuming hungry agents) & $\checkmark$ (Thm~\ref{thm:rat_p_risk-averse_hungry}) & $\checkmark$ (Thm~\ref{thm:rat_p_risk-averse_hungry}) & $\checkmark$ (Thm~\ref{thm:maximinSP_proportional}) & $\checkmark$ (Thm~\ref{thm:rat_p_nom}) \\
    Mechanism~\ref{alg:rat_p'} & $\checkmark$ (Thm~\ref{thm:rat_p_change}) & $\checkmark$ (Thm~\ref{thm:rat_p_change}) & $\checkmark$ (Thm~\ref{thm:maximinSP_proportional}) & $\times$ (Thm~\ref{thm:rat_p'_nom}) \\
    \bottomrule
    \end{tabular}
    \caption{Mechanisms and their satisfabilities to the truthful notions proportional risk-averse truthfulness (PRAT), risk-averse truthfulness (RAT), maximin strategy-proofness (maximinSP), and not obvious manipulability (NOM).\\
    $\dag$ The risk-averse truthfulness for the Even-Paz algorithm and Ortega and Segal-Halevi's moving-knife procedure depend on some subtle tie-breaking issues. Both mechanisms are risk-averse truthful if agents are hungry. \ref{append:EvenPaz} and \ref{append:OSRAT} elaborate these.}
    \label{tab:mechanisms}
\end{table}

Finally, we remark that a common Bayesian model captures the uncertainty of other agents' private information: define a probability distribution from which an agent believes that the other agents' private information is drawn (typically, this distribution depends on the information this agent has).
This is a typical setting in the auction theory (e.g., an agent believes that another agent's valuation on an item is drawn uniformly at random from $[0,1]$).
However, in our case, we do not see any natural way to define a probability distribution over piecewise-constant functions.

\section{Risk-Averse Truthful Envy-Free Mechanisms}
\label{sect:RAT_EF}
There exists a simple algorithm that outputs envy-free allocations for $n$ agents with piecewise-constant value density functions.
The algorithm first collects all the points of discontinuity from all agents.
This partitions the cake into multiple intervals where each agent's value density function is uniform on each of these intervals.
Then, the algorithm uniformly allocates each interval to all agents.
The output allocation $(A_1,\ldots,A_n)$ of this algorithm satisfies $v_i(A_j)=\frac1n v_i([0,1])$ (this property of an allocation is called \emph{perfect}), which is clearly envy-free.
However, to make the algorithm deterministic, we need to specify a left-to-right order of the $n$ agents on how each interval is allocated.
The algorithm is described in Mechanism~\ref{alg:simple_ef}.

\begin{algorithm}
\caption{A simple envy-free cake cutting algorithm}
\label{alg:simple_ef}
\begin{algorithmic}[1]
\STATE let $X_i$ be the set of all points of discontinuity for $f_i$\;
\STATE let $X=\bigcup_{i=1}^nX_i$\;
\STATE let $X=\{x_1,\ldots,x_{m-1}\}$ be sorted by ascending order, and let $x_0=0,x_m=1$\;
\STATE initialize $A_i=\emptyset$ for each $i=1,\ldots,n$\;
\STATE \textbf{for} each $j=0,1,\ldots,m-1$:
\STATE \hspace{0.5cm} for each agent $i=1,\ldots,n$: $A_i\leftarrow A_i\cup \left[x_j+\frac{i-1}n(x_{j+1}-x_j),x_j+\frac in(x_{j+1}-x_j)\right)$;
\STATE \textbf{endfor}
\STATE \textbf{return} allocation $(A_1,\ldots,A_n)$\;
\end{algorithmic}
\end{algorithm}

However, Mechanism~\ref{alg:simple_ef} is not even risk-averse truthful.

\begin{theorem}
Mechanism~\ref{alg:simple_ef} is not risk-averse truthful.
\end{theorem}
\begin{proof}
Consider $f_1$ such that $f_1(x)=1$ for $x\in[0,\frac1n)$ and $f_1(x)=0.5$ for $x\in[\frac1n,1]$, and consider $f_1'(x)=1$ for $x\in[0,1]$.
Let $\M$ be the mechanism.
We aim to show that, 1) there exist $f_2,\ldots,f_n$ such that $v_1(\M_1(f_1',f_2,\ldots,f_n))>v_1(\M_1(f_1,f_2,\ldots,f_n))$, and 2) for any $f_2,\ldots,f_n$, $v_1(\M_1(f_1',f_2,\ldots,f_n))\geq v_1(\M_1(f_1,f_2,\ldots,f_n))$.
That is, misreporting $f_1$ to $f_1'$ is sometimes more beneficial and always no harm.

To show 1), consider $f_2(x)=\cdots=f_n(x)=1$ for $x\in[0,1]$.
If agent $1$ truthfully reports $f_1$, (s)he will receive $[0,\frac1{n^{2}})\cup[\frac1n,\frac1n+\frac{n-1}{n^{2}})$, which is worth $\frac1{n^2}+\frac{n-1}{2n^2}$.
If agent $1$ reports $f_1'$, the mechanism will see $n$ uniform functions, and allocation $[0,\frac1n)$ to agent $1$, which is worth $\frac1n$, which is more than $\frac1{n^2}+\frac{n-1}{2n^2}$.

To show 2), consider any $f_2,\ldots,f_n$.
Suppose agent $1$ reports $f_1'$. Let $X$ be defined in Step~2 and 3 of the mechanism with respect to $f_1',f_2,\ldots,f_n$.
Agent $1$ always receives the leftmost $1/n$ fraction of each $[x_j,x_{j+1})$.
Since $f_1$ is monotonically decreasing, this is worth at least $1/n$ of $v([x_j,x_{j+1}))$, and agent $1$ receives at least his/her proportional share overall.
On the other hand, if agent $1$ truthfully reports $f_1$, (s)he will always receive exactly his/her proportional share, which is weakly less than what (s)he would receive by reporting $f_1'$.
\end{proof}

The reason for Mechanism~\ref{alg:simple_ef} not being risk-averse truthful is that an agent can ``delete'' a point of discontinuity to merge two intervals $[x_j,x_{j+1})$ and $[x_{j+1},x_{j+2})$.
This may be more beneficial if his/her value is higher on $[x_j,x_{j+1})$ (or $[x_{j+1},x_{j+2})$) and (s)he knows that the mechanism will allocate a piece on the very left (or very right) of $[x_j,x_{j+2})$.
Therefore, it is the deterministic left-to-right order on each interval that compromises the truthfulness.
It is easy to randomize Mechanism~\ref{alg:simple_ef} such that Mechanism~\ref{alg:simple_ef} is \emph{truthful in expectation}, meaning that an expected utility optimizing agent's dominant strategy is truth-telling.
To achieve this, we just need to partition each $[x_j,x_{j+1})$ evenly into $n$ pieces and allocate these $n$ pieces to the $n$ agents by a random perfect matching.
This is essentially the Mechanism proposed by Mossel and Tamuz~\cite{mossel2010truthful}.

We propose a deterministic proportionally risk-averse truthful and envy-free mechanism that uses similar ideas.
The mechanism is the same as Mechanism~\ref{alg:simple_ef}, except that the left-to-right order on each interval $[x_j,x_{j+1})$ depends on the index $j$.
Intuitively, if an agent tries to merge two intervals, (s)he does not know where exactly his/her $1/n$ fraction of $[x_j,x_{j+1})$ is, as (s)he does not know other agents' value density functions.
This makes it possible that (s)he ends up receiving a portion where (s)he has less value on.
The mechanism is shown in Mechanism~\ref{alg:rat_ef}.

\begin{algorithm}
\caption{A risk-averse truthful envy-free cake cutting mechanism}
\label{alg:rat_ef}
\begin{algorithmic}[1]
\STATE let $X_i$ be the set of all points of discontinuity for $f_i$\;
\STATE let $X=\bigcup_{i=1}^nX_i$\;
\STATE let $X=\{x_1,\ldots,x_{m-1}\}$ be sorted by ascending order, and let $x_0=0,x_m=1$\;
\STATE initialize $A_i=\emptyset$ for each $i=1,\ldots,n$\;
\STATE \textbf{for} each $j=0,1,\ldots,m-1$:
\STATE \hspace{0.5cm} for each agent $i$: $A_i\leftarrow A_i\cup \left[x_j+\frac{i+j-1\mod n}n(x_{j+1}-x_j),x_j+\frac{(i+j-1\mod n) +1}n(x_{j+1}-x_j)\right)$;
\STATE \textbf{endfor}
\STATE \textbf{return} allocation $(A_1,\ldots,A_n)$\;
\end{algorithmic}
\end{algorithm}

\begin{theorem}\label{thm:rat_ef}
Mechanism~\ref{alg:rat_ef} is proportionally risk-averse truthful and envy-free.
\end{theorem}
\begin{proof}
The envy-freeness is trivial. We will focus on proportional risk-averse truthfulness.
The part of proportionality is also trivial, as an entire envy-free allocation is always proportional and Mechanism~\ref{alg:rat_ef} is entire.

We focus on agent $1$ without loss of generality.
Let $f_1$ be agent $1$'s true value density function.
Consider an arbitrary $f_1'$ that agent $1$ reports.
Let $X_1$ and $X_1'$ be the sets of all points of discontinuity for $f_1$ and $f_1'$ respectively.

Suppose $X_1\subseteq X_1'$.
It is easy to see that agent $1$ will still get a value of $\frac1nv_1([0,1])$ by reporting $f_1'$.
This is because any subdivision of an interval where agent $1$ has a uniform value gives only smaller intervals each of which agent $1$ has a uniform value on.
This kind of misreporting is captured by 2(a) of Definition~\ref{def:rat}.

Suppose $X_1\not\subseteq X_1'$.
Pick an arbitrary $t\in X_1\setminus X_1'$.
Assume without loss of generality that $\displaystyle\lim_{x\rightarrow t^-}f(x)<\lim_{x\rightarrow t^+}f(x)$.
Consider a sufficiently small $\varepsilon>0$ such that $[t-\varepsilon,t+(n-1)\varepsilon]$ do not contain any points in $X_1\cup X_1'\setminus\{t\}$.
We can construct $f_2,\ldots,f_n$ such that 1) $\bigcup_{i=2}^nX_i$ contains $X_1\cup X_1' \cup\{t-\varepsilon, t+(n-1)\varepsilon\}\setminus\{t\}$, 2) $\bigcup_{i=2}^nX_i$ do not intersect the open interval $(t-\varepsilon,t+(n-1)\varepsilon)$, and 3) $t-\varepsilon$ is the $j$-th point from left to right with $j$ being a multiple of $n$.
By our mechanism, agent $1$ will receive $[t-\varepsilon,t)$ on the $j$-th interval $[t-\varepsilon,t+(n-1)\varepsilon)$, which is worth less than $\frac1nv_1([t-\varepsilon,t+(n-1)\varepsilon))$.
Agent $1$ will receive value exactly $\frac1nv_1([0,1]\setminus [t-\varepsilon,t+(n-1)\varepsilon))$ on the remaining part of the cake.
Therefore, the overall value agent $1$ receives is below the proportional value.
We have shown that this type of misreporting may cause agent $1$'s received value to be less than the proportional value, which corresponds to 2(b) of Definition~\ref{def:rat}.
\end{proof}

\section{Risk-Averse Truthful Proportional Mechanisms with Connected Pieces}
\label{sect:RAT_P}
We have seen that Mechanism~\ref{alg:rat_ef} is proportionally risk-averse truthful.
However, each agent may receive a union of quite many intervals instead of a single interval.
This is undesirable in many applications where people want a contiguous piece of resource, e.g., dividing a piece of land, allocating meeting time slots.
In this section, we are looking for proportionally risk-averse truthful mechanisms that satisfy the \emph{connected pieces} property.
That is, we require that each agent must receive a connected interval of the cake.

Many existing algorithms output proportional allocations with connected pieces.
Two notable algorithms are \emph{the moving-knife procedure}~\cite{dubins1961cut} and \emph{the Even-Paz algorithm}~\cite{even1984note}.
We will see in this section that both algorithms are not proportionally risk-averse truthful.
In particular, the moving-knife procedure is not even risk-averse truthful.
We conclude this section by proposing a proportionally risk-averse truthful  mechanism with connected pieces.

\paragraph{Moving-knife procedure}
Let $a_i=\frac1nv_i([0,1])$ be agent $i$'s proportional value.
The moving-knife procedure marks for each agent $i$ a point $x_i$ such that $[0,x_i)$ is worth exactly $a_i$ to agent $i$.
Then, the algorithm finds the smallest value $x_{i^\ast}$ among $x_1,\ldots,x_n$, and allocates $[0,x_{i^\ast})$ to agent $i^\ast$.
Next, for the remaining part of the cake $[x_{i^\ast},1]$, the algorithm marks for each of the $n-1$ remaining agents a point $x_i'$ such that $[x_{i^\ast},x_i')$ is worth exactly $a_i$ to agent $i$.
The algorithm then finds the smallest value $x_{i^\dag}$ among those $n-1$ $x_i'$s, and allocates $[x_{i^\ast},x_{i^\dag})$ to agent $i^\dag$.
This is repeated until the $(n-1)$-th agent is allocated an interval, and then the last agent gets the remaining part of the cake.
It is easy to verify that each of the first $n-1$ agents receives an interval that is worth exactly his/her proportional value $a_i$, while the last agent may receive more than his/her proportional value.

\paragraph{Even-Paz algorithm}
The Even-Paz algorithm is a divide-and-conquer-based algorithm.
For each agent $i$, Even-Paz algorithm finds a point $x_i$ such that $v_i([0,x_i])=\lfloor\frac n2\rfloor \frac1nv_i([0,1])$.
It then find the median $x^\ast$ for $x_1,\ldots,x_n$.
Let $L$ be the set of agents $i$ with $x_i<x^\ast$ and $R$ be the set of agents $i$ with $x_i\geq x^\ast$.
Since each agent $i$ in $L$ believes $v_i([0,x^\ast])\geq \lfloor\frac n2\rfloor\frac1n v_i([0,1])$ and there are $\lfloor\frac n2\rfloor$ agents in $L$, there exists an allocation of $[0,x^\ast]$ to agents in $L$ such that each agent $i$ receives at least his/her proportional value $\frac1n v_i([0,1])$.
For the similar reasons, there exists an allocation of $(x^\ast,1]$ to agents in $R$ such that each agent $i$ receives at least his/her proportional value $\frac1n v_i([0,1])$.
The algorithm then solves these two problems recursively.
It is also easy to prove that the Even-Paz algorithm always outputs proportional allocations.

\bigskip

To show that both algorithms are not proportionally risk-averse truthful.
We first define the following two value density functions.
\begin{equation}\label{eqn:lr}
    \ell^{(n)}(x)=\left\{\begin{array}{ll}
        \frac32 & x\in\left[0,\frac1{2n}\right) \\
        \frac12 & x\in\left[\frac1{2n},\frac1n\right)\\
        1 & x\in\left[\frac1n,1\right]
    \end{array}\right.\qquad
    r^{(n)}(x)=\left\{\begin{array}{ll}
        1 & x\in\left[0,1-\frac1n\right)\\
        \frac12 & x\in\left[1-\frac1n,1-\frac1{2n}\right)\\
        \frac32 & x\in\left[1-\frac1{2n},1\right]
    \end{array}\right.
\end{equation}

Notice that $\int_0^1\ell^{(n)}(x)dx=\int_0^1r^{(n)}(x)dx=1$.
The following lemma shows that any connected allocation that is proportional in either $\ell^{(n)}$ or $r^{(n)}$ is also proportional in the uniform value density function.

\begin{lemma}\label{lem:lr}
Let $f(x)=1$ for $x\in[0,1]$. For any interval $I$ such that $\int_I\ell^{(n)}(x)dx\geq\frac1n$, we have $\int_If(x)dx\geq\frac1n$. For any interval $I$ such that $\int_Ir^{(n)}(x)dx\geq\frac1n$, we have $\int_If(x)dx\geq\frac1n$.
\end{lemma}
\begin{proof}
We only prove the lemma for $\int_I\ell^{(n)}(x)dx\geq\frac1n$, as the proof for $\int_Ir^{(n)}(x)dx\geq\frac1n$ is similar.
It is straightforward to see that $\int_I\ell^{(n)}(x)dx=\frac1n$ implies $|I|\geq\frac1n$.
In particular, $|I|=\frac1n$ if the left endpoint of $I$ belongs to $\{0\}\cup[\frac1n,1-\frac1n]$, and $|I|>\frac1n$ if the left endpoint of $I$ belongs to $(0,\frac1n)$.
For $|I|\geq\frac1n$, we have $\int_If(x)dx\geq\frac1n$.
If $\int_I\ell^{(n)}(x)dx>\frac1n$, there exists $I'\subseteq I$ such that $\int_{I'}\ell^{(n)}(x)dx=\frac1n$.
By our previous analysis, $|I'|\geq\frac1n$. 
We have $\int_If(x)dx\geq \int_{I'}f(x)dx\geq\frac1n$.
\end{proof}

\begin{theorem}\label{thm:moving-knife}
The moving-knife procedure is not risk-averse truthful.
\end{theorem}
\begin{proof}
Let $f_1(x)=1$ for $x\in[0,1]$ be the true value density function for agent~$1$.
We show that agent~$1$ can misreport his/her value density function to $f_1'=\ell^{(n)}$ that satisfies 1) there exist $f_2,\ldots,f_n$ such that $v_1(\M_1(f_1',f_2,\ldots,f_n))>v_1(\M_1(f_1,f_2,\ldots,f_n))$, and 2) for any $f_2,\ldots,f_n$, $v_1(\M_1(f_1',f_2,\ldots,f_n))\geq v_1(\M_1(f_1,f_2,\ldots,f_n))$.

To see 1), suppose $f_2(x)=1$ for $x\in[0,\frac1n]$ and $f_2(x)=0$ for $x\in(\frac1n,1]$, and $f_3(x)=\cdots=f_n(x)=0$ for $x\in[0,\frac1n)$ and $f_3(x)=\cdots=f_n(x)=1$ for $x\in[\frac1n,1]$.
In the moving-knife procedure, if agent $1$ truthfully reports $f_1$, (s)he will be the second agent receiving an interval after agent $2$ taking $[0,\frac1{n^2})$, and (s)he will receive $[\frac1{n^2},\frac1n+\frac1{n^2})$, which is worth $\frac1n$.
If agent $1$ reports $f_1'$, (s)he will also be the second agent receiving an interval after agent $2$ taking $[0,\frac1{n^2})$, and (s)he will receive $[\frac1{n^2},\frac1n+\frac3{2n^2})$ (by some simple calculations), which is worth more than $\frac1n$ with respect to his/her true valuation.

To see 2), suppose agent $1$ reports $f_1'$.
Since the moving-knife procedure is proportional, regardless of what the remaining $n-1$ agents report, agent $1$ will receive an interval that has a value of at least $\frac1n$ with respect to $f_1'$.
By Lemma~\ref{lem:lr}, agent $1$ receives an interval that is worth at least $\frac1n$ with respect to his/her true valuation $f_1$.
This already shows that the moving-knife procedure is not proportionally risk-averse truthful.

We can further show that the procedure is not even risk-averse truthful.
Consider any $f_2,\ldots,f_n$. 
If agent $1$ is not the last agent receiving an interval by reporting $f_1$ truthfully, agent $1$ receives exactly value $\frac1n$ by the nature of the moving-knife procedure.
Since we have shown that reporting $f_1'$ also guarantees the proportionality of agent $1$, reporting $f_1'$ will not harm agent $1$.
Suppose agent $1$ is the last agent receiving an interval by reporting $f_1$ truthfully.
Now, suppose agent $1$ reports $f_1'$.
In each iteration of the procedure, by Lemma~\ref{lem:lr}, agent $1$'s marked point for reporting $f_1'$ is the same as, or on the right-hand side of, agent $1$'s marked point for reporting $f_1$.
This indicates that agent $1$ will still be the last agent to receive an interval when reporting $f_1'$.
Moreover, the first $n-1$ points cut by the procedure will only depend on $f_2,\ldots,f_n$.
Thus, when agent $1$ reports $f_1'$, agent $1$ receives the same interval as it is in the case where agent $1$ reports $f_1$.
In this case, reporting $f_1'$ does not harm agent $1$ as well.
\end{proof}

\begin{theorem}\label{thm:even-paz}
The Even-Paz algorithm is not proportionally risk-averse truthful.
\end{theorem}
\begin{proof}
Consider the scenario with $n=5$ agents.
Let $f_1(x)=1$ for $x\in[0,1]$ be the true value density function for agent~$1$.
We show that agent~$1$ can misreport his/her value density function to $f_1'=r^{(5)}$ that satisfies 1) there exist $f_2,f_3,f_4,f_5$ such that $v_1(\M_1(f_1',f_2,f_3,f_4,f_5))>v_1(\M_1(f_1,f_2,f_3,f_4,f_5))$, and 2) for any $f_2,f_3,f_4,f_5$, we have $v_1(\M_1(f_1',f_2,f_3,f_4,f_5))\geq \frac15v_1([0,1])$.
Since the Even-Paz algorithm is proportional, Lemma~\ref{lem:lr} immediately implies 2). 
It remains to show 1).

Let $\varepsilon>0$ be a small number less than $\frac1{10}$. 
Consider $f_2(x)=1$ on $[0,\varepsilon)$ and $f_2(x)=0$ on $[\varepsilon,1]$, and $f_3(x)=f_4(x)=f_5(x)=0$ on $[0,1-\varepsilon)$ and $f_3(x)=f_4(x)=f_5(x)=1$ on $[1-\varepsilon,1]$.
We analyze two cases: the case where agent $1$ truthfully reports $f_1$ and the case where agent $1$ reports $f_1'$.
It is easy to verify that, in both cases, after the first round of the algorithm, an allocation of $[0,1-\frac35\varepsilon]$ to agent $1$ and $2$ is to be decided, and an allocation of $(1-\frac35\varepsilon,1]$ to agent $3,4,5$ is to be decided.
In the next round, the algorithm will find the half-half point for each of agent $1$ and $2$ on $[0,1-\frac35\varepsilon]$, and the algorithm will cut at the median of the two points, which is the average of the two points, and allocate the right-hand side interval to agent $1$.
By some simple calculations, the half-half point of $f_1$ on $[0,1-\frac35\varepsilon]$ is to the right of the half-half point of $f_1'$ on $[0,1-\frac35\varepsilon]$.
As a result, agent $1$ will receives a larger length of interval if (s)he reports $f_1'$.
Since the true value density function $f_1$ is uniform, reporting $f_1'$ will give agent $1$ more utility.
\end{proof}

In \ref{append:EvenPaz}, we will see that the Even-Paz algorithm can be made risk-averse truthful if some subtle tie-breaking issues are handled correctly.

To conclude this section, we present a mechanism that is proportionally risk-averse truthful.
In particular, if we require the entire allocations, it is proportionally risk-averse truthful for hungry agents.
The mechanism is shown in Mechanism~\ref{alg:rat_p}.
Later, we will show that we can modify the mechanism by a little bit to make it proportionally risk-averse truthful (without assuming the agents are hungry) if we do not require entire allocations (while still guaranteeing proportionality and connected pieces).

\begin{algorithm}
\caption{A proportionally risk-averse truthful cake cutting mechanism with connected pieces}
\label{alg:rat_p}
\begin{algorithmic}[1]
\STATE for each $f_i$, find the smallest $x^{(i)}_1,\ldots,x^{(i)}_{n-1}$ such that $\int_{x_j^{(i)}}^{x_{j+1}^{(i)}}f_i(x)dx=\frac1n\int_0^1f_i(x)dx$ for each $j=0,1,\ldots,n-1$, where $x^{(i)}_0=0$ and $x^{(i)}_n=1$\;
\STATE $c_0\leftarrow 0$\;
\STATE $\text{Unallocated}\leftarrow\{1,\ldots,n\}$\; // the set of agents who have not been allocated
\STATE \textbf{for} each $j=1,\ldots,n-1$:
\STATE \hspace{0.5cm} $i_j\leftarrow\arg\min_{i\in\text{Unallocated}}\{x_j^{(i)}\}$\;
\STATE \hspace{0.5cm} $c_j\leftarrow x_{j}^{(i_j)}$\;
\STATE \hspace{0.5cm} allocate $[c_{j-1},c_j)$ to agent $i_j$\;
\STATE \hspace{0.5cm} $\text{Unallocated}\leftarrow \text{Unallocated}\setminus\{i_j\}$\;
\STATE \textbf{endfor}
\STATE allocate the remaining unallocated interval to the one remaining agent in Unallocated.
\end{algorithmic}
\end{algorithm}

\begin{theorem}\label{thm:rat_p_proportional}
Mechanism~\ref{alg:rat_p} is entire and proportional, and it always outputs allocations with connected pieces.
\end{theorem}
\begin{proof}
It is trivial that the mechanism is entire and always outputs allocations with connected pieces. It remains to show the proportionality.
It suffices to show that, in each iteration $j$, we have $[x^{(i_j)}_{j-1},x^{(i_j)}_{j})\subseteq[c_{j-1},c_j)$ (notice that $[x^{(i_j)}_{j-1},x^{(i_j)}_{j})$ is worth exactly the proportional value for agent $i_j$).
Since $x^{(i_j)}_{j}=c_j$, it suffices to show that $x^{(i_j)}_{j-1}\geq c_{j-1}$.
In the $(j-1)$-th iteration, agent $i_j$ is still in the set Unallocated.
Since $i_{j-1}$ is the agent $i$ in Unallocated with minimum $x_{j-1}^{(i)}$, we have $x^{(i_j)}_{j-1}\geq x^{(i_{j-1})}_{j-1}= c_{j-1}$.
\end{proof}

\begin{theorem}\label{thm:rat_p_risk-averse_hungry}
Mechanism~\ref{alg:rat_p} is proportionally risk-averse truthful for hungry agents.
\end{theorem}
\begin{proof}
Without loss of generality, we consider the potential misreport for agent $1$.
Let $f_1$ be agent $1$'s true value density function, and consider an arbitrary $f_1'$.
If the values for $x_1^{(1)},\ldots,x_{n-1}^{(1)}$ (in Step~1 of the mechanism) are the same for $f_1$ and $f_1'$, the mechanism will output the same allocation for $f_1$ and $f_1'$.
In this case, reporting $f_1'$ is not strictly more beneficial.
We will conclude the proof by showing that, if the values for $x_1^{(1)},\ldots,x_{n-1}^{(1)}$ are not the same for $f_1$ and $f_1'$, there exist $f_2,\ldots,f_n$ such that agent $1$ will receive an interval with value less than the proportional value (with respect to the true valuation $f_1$).

Suppose $j^\ast$ is the minimum index such that $x_{j^\ast}^{(1)}$ is not the same for $f_1$ and $f_1'$.
Let $y$ be the value of $x_{j^\ast}^{(1)}$ for $f_1$ and $y'$ be the value of $x_{j^\ast}^{(1)}$ for $f_1'$.
We consider two cases: $y'<y$ and $y'>y$.
Let $\varepsilon>0$ be a sufficiently small number.

Suppose $y'<y$.
We can construct $f_2,\ldots,f_n$ such that 1) for each $j=1,\ldots,j^\ast-1$, $c_j=x_{j}^{(1)}-\varepsilon$, and 2) $c_{j^\ast}=y'$.
In this case, agent $1$ will receive $[x_{j^\ast-1}^{(1)}-\varepsilon,y')$.
When $\varepsilon\rightarrow0$, this interval converges to $[x_{j^\ast-1}^{(1)},y']$, which is a proper subset of $[x_{j^\ast-1}^{(1)},y)$.
We know that $[x_{j^\ast-1}^{(1)},y)$ is just enough to guarantee the proportionality for agent $1$.
Agent $1$ receives an interval with a value less than the proportional value by reporting $f_1'$, if $\varepsilon$ is small enough.

Suppose $y'>y$.
Since each of the intervals $[x_0^{(1)},x_1^{(1)}),\ldots,[x_{j^\ast-2}^{(1)},x_{j^\ast-1}^{(1)})$ is worth exactly $\frac1n v_1([0,1])$ and the interval $[x_{j^\ast-1}^{(1)},y')$ is worth strictly more than $\frac1n v_1([0,1])$, the interval $[y',1]$ is worth less than $\frac{n-j^\ast}n v_1([0,1])$.
It is possible to find $y_{j^\ast+1},\ldots,y_{n-1}$ such that $[y_j,y_{j+1})$ is worth strictly less than $\frac1nv_1([0,1])$ for each $j=j^\ast,\ldots,n-1$, where we let $y_{j^\ast}=y'$ and $y_n=1$.
Now we construct $f_2,\ldots,f_n$ such that 1) $c_j=x_j^{(1)}-\varepsilon$ for each $j=1,\ldots,j^\ast-1$, 2) $c_{j^\ast}=y'-\varepsilon$, and 3) $\min_{i}x_j^{(i)}=y_j$ for each $j=j^\ast+1,\ldots,n-1$.
It is easy to see that agent $1$ will receive an interval that is a subset of one of $[y_{j^\ast},y_{j^\ast+1}),\ldots,[y_{n-1},1]$.
Therefore, agent $1$ will receive a value less than the proportional value in this case.
\end{proof}

If the agents are not hungry, the set of points $x_1^{(i)},\ldots,x_{n-1}^{(i)}$ satisfying the condition in Step~1 may not be unique.
Different selections of this set may result in different allocations.
An agent can select this set (by reporting an $f_i'$ with $x_1^{(i)},\ldots,x_{n-1}^{(i)}$ being exactly what (s)he want) and potentially receive a better allocation.

It is possible to get rid of the hungry agents assumption.
The trick is to make sure that each agent $i$ receives \emph{exactly} one of $[0,x_1^{(1)}),[x_1^{(1)},x_2^{(1)}),\ldots,[x_{n-1}^{(1)},1]$.
We only need to change Step~7 of Mechanism~\ref{alg:rat_p} to ``allocate $[x_{j-1}^{(i_j)},c_j)$ to agent $i_j$''.
The mechanism is stated in Mechanism~\ref{alg:rat_p'}.

\begin{algorithm}
\caption{A proportionally risk-averse truthful cake cutting mechanism with connected pieces}
\label{alg:rat_p'}
\begin{algorithmic}[1]
\STATE for each $f_i$, find the smallest $x^{(i)}_1,\ldots,x^{(i)}_{n-1}$ such that $\int_{x_j^{(i)}}^{x_{j+1}^{(i)}}f_i(x)dx=\frac1n\int_0^1f_i(x)dx$ for each $j=0,1,\ldots,n-1$, where $x^{(i)}_0=0$ and $x^{(i)}_n=1$\;
\STATE $c_0\leftarrow 0$\;
\STATE $\text{Unallocated}\leftarrow\{1,\ldots,n\}$\; // the set of agents who have not been allocated
\STATE \textbf{for} each $j=1,\ldots,n-1$:
\STATE \hspace{0.5cm} $i_j\leftarrow\arg\min_{i\in\text{Unallocated}}\{x_j^{(i)}\}$\;
\STATE \hspace{0.5cm} $c_j\leftarrow x_{j}^{(i_j)}$\;
\STATE \hspace{0.5cm} allocate $[x_{j-1}^{(i_j)},c_j)$ to agent $i_j$\;
\STATE \hspace{0.5cm} $\text{Unallocated}\leftarrow \text{Unallocated}\setminus\{i_j\}$\;
\STATE \textbf{endfor}
\STATE allocate the remaining unallocated interval to the one remaining agent in Unallocated.
\end{algorithmic}
\end{algorithm}

In this case, as long as an agent selects a set $x_1^{(i)},\ldots,x_{n-1}^{(i)}$ that satisfies the condition in Step~1, (s)he will get exactly his/her proportional share.
Of course, if (s)he selects a set $x_1^{(i)},\ldots,x_{n-1}^{(i)}$ that does not satisfy the condition, the same arguments in the proof of Theorem~\ref{thm:rat_p_risk-averse_hungry} show that there is always a scenario that (s)he will receive a value less than the proportional value.
These prove the theorem below, which is stated with the formal proof left to the readers.

\begin{theorem}\label{thm:rat_p_change}
Mechanism~\ref{alg:rat_p'} is proportionally risk-averse truthful  (but not entire).
\end{theorem}

However, compared with Mechanism~\ref{alg:rat_p}, other than not being entire, another disadvantage of Mechanism~\ref{alg:rat_p'} is that it is obviously manipulable (see \ref{append:nom} and Theorem~\ref{thm:rat_p'_nom}).

\subsection{Remarks on Computational Complexity}
Although computational complexity of mechanisms is not the main focus of this paper (in fact, the impossibility results in Sect.~\ref{sect:impossibility} is irrelevant to computational complexity, and they also exclude the possibility of super-polynomial time mechanisms), being able to be executed in a polynomial time is still a desirable property for a practical mechanism.
It is easy to check that our mechanisms in this section, as well as the one in the previous section, can be implemented in polynomial time (in terms of the length of the string encoding all the $n$ value density functions, as it is standard in complexity theory).



\section{Conclusion and Future Work}
\label{sect:conclusion}
We have proved that a truthful proportional cake cutting mechanism does not exist, even in the restrictive setting with two agents whose value density functions are piecewise-constant and strictly positive.
The impossibility result extends to the setting where it is not required that the entire cake needs to be allocated.
This resolves the long-standing fundamental open problem in the cake cutting literature.
The main take-home message for this paper is that dominant-strategy truthfulness and fairness cannot be both guaranteed for the cake cutting problem.
Therefore, to deploy a cake-cutting mechanism, we need to further relax dominant-strategy truthfulness or fairness.

\subsection{Relaxing Truthfulness}
For relaxing dominant-strategy truthfulness, we have proposed a new truthful notion called \emph{(proportionally) risk-averse truthfulness}, which is motivated by the truthful property that the I-cut-you-choose mechanism possesses.
We have shown that some well-known cake cutting algorithms do not satisfy this truthful criterion.
We have provided a proportionally risk-averse truthful and envy-free mechanism and a proportionally risk-averse truthful mechanism that always outputs allocations with connected pieces.

In some scenarios where randomized mechanisms are acceptable and agents are generally risk-neutral, another option is the randomized mechanism proposed by Mossel and Tamuz~\cite{mossel2010truthful} that is truthful in expectation.

\subsection{Relaxing Proportionality}
On the other hand, we can relax the proportionality requirement, and instead, consider the approximation of proportionality.
We have seen in Theorem~\ref{thm:main2_approx} that there does not exist a truthful and $0.974031$-approximately proportional mechanism.
How about smaller approximation ratios?

\begin{openP}
Does there exist an $\alpha>0$ such that there exists a truthful, $\alpha$-approximately proportional mechanism?
\end{openP}

Designing dominant-strategy truthful mechanisms for piecewise-constant value density functions is still a largely unexplored research area.
To the best of our knowledge, there is no ``natural'' dominant-strategy truthful mechanism if agents' value density functions are piecewise-constant.
We only know some ``unnatural'' truthful mechanisms that either are oblivious to one or more agents' valuation (e.g., allocate the whole cake to a fixed single agent, allocate the cake evenly to $n$ agents such that each agent receives a length of $\frac1n$ disregarding agents' valuations, etc), or cannot even guarantee each agent a positive value (e.g., the mechanism can arbitrarily fix two different allocations $(A_1,\ldots,A_n)$ and $(A_1',\ldots,A_n')$ and let the $n$ agents vote for the more preferred allocation; this mechanism is truthful and non-oblivious to all agents' valuations, but some agents may receive pieces with a zero value).
These mechanisms cannot guarantee even the minimum level of fairness.

Indeed, we do not even know the existence of a truthful mechanism that guarantees each agent a positive value.
If the answer to the following open problem is no, we have the same impossibility result as the result of Br{\^a}nzei and Miltersen~\cite{branzei2015dictatorship} for the Robertson-Webb query model.

\begin{openP}
Does there exist a truthful mechanism that always allocates each agent a subset on which the agent has a positive value?
\end{openP}

Of course, if agents are hungry, the answer to the problem above is yes, as the mechanism can just allocate $[0,1]$ to the agents such that each agent receives a length of $\frac1n$, disregarding the agents' reports.

In conclusion, designing a ``reasonable'' truthful mechanism is still a challenging problem.

\subsection{Special Value Density Functions}
Although our main result shows that truthfulness and proportionality are incompatible for general piecewise-constant value density functions, they may be compatible for some \emph{special cases} of piecewise-constant functions such as piecewise-uniform functions (reference~\cite{CL10}) and monotone functions (Sect.~\ref{sect:monotone}).
Another direction for future work is to figure out for what subsets of value density functions we can have a truthful and proportional mechanism.

In many practical scenarios, the precision of value density functions can be limited.
For example, there is a small constant $k$ such that value density functions can only take value from $\{0,1/k,2/k,\ldots,1\}$.
The case for $k=1$ corresponds to piecewise-uniform functions, in which case a truthful, envy-free and entire mechanism exists even for any number of agents~\cite{CL10}.
The theorem below, whose proof is deferred to \ref{append:restricted}, shows that our impossibility result continues to hold for $k\geq3$.

\begin{theorem}\label{thm:main_restricted}
There does not exist a truthful proportional mechanism, even when all of the following hold:
\begin{itemize}
    \item there are two agents;
    \item each agent's value density function is piecewise-constant and only takes value from $\{0,1/k,2/k,\ldots,1\}$ for any fixed $k\geq 3$;
    \item the mechanism needs not to be entire.
\end{itemize}
\end{theorem}

We do not know if the impossibility result holds for $k=2$.
In general, characterizing the subsets of value density functions where a truthful and proportional (or envy-free) mechanism exists is an interesting future research direction.

\subsection{Cake Cutting with More Than Two Agents}
We have proved the impossibility result on truthful proportional mechanisms with $n=2$.
Although this implies such mechanisms do not exist in general, it still makes sense to consider this problem with a fixed number of agents that is more than $2$.
We conjecture that the impossibility result holds for any fixed $n\geq2$.

\begin{openP}
Does there exist a positive integer $n\geq3$ such that there exists a truthful proportional mechanism with $n$ agents?
\end{openP}

A common technique for extending the impossibility result from $n=2$ to a general fixed $n$ is to add $n-2$ dummy agents who only have positive values on $n-2$ non-intersecting intervals that do not intersect with the valued intervals of the first two agents.
However, this technique fails here, as adding more agents would reduce the proportional guarantee from $\frac12v_i([0,1])$ to $\frac1nv_i([0,1])$.
In particular, an allocation where agent $1$ and $2$ get exactly $1/n$ of their values on the entire cake (while the remaining parts of agent $1$'s and agent $2$'s valued intervals are allocated to the remaining $n-2$ agents, even if the remaining $n-2$ agents have a zero value on these intervals) can still be proportional, or even envy-free.
This allocation, on the other hand, is not proportional if we turn back to the instance with two agents (with the dummy agents removed).

\subsection{Empirical Studies}
We have proposed two mechanisms that are risk-averse truthful.
It is also interesting to test them empirically by simulations or sociological experiments and compare their performances with other classical algorithms such as the moving-knife procedure and the Even-Paz algorithm.

\section*{Acknowledgments}
The authors would like to thank Xiaohui Bei, Grant Schoenebeck, and the anonymous reviewers for their great suggestions on this paper.

The research of Biaoshuai Tao was supported by the National Natural Science Foundation of China (Grant No. 62102252).





\bibliographystyle{elsarticle-num} 
\bibliography{reference}

%




\newpage
\appendix

\setcounter{instance}{0}

\section{Proof of Theorem~\ref{thm:main2_approx}}
\label{append:proof_approx}
Let $\M$ be a truthful and $(1-\tau)$-approximately proportional mechanism for certain $\tau\in[0,0.025969]$.
Like the proof for Theorem~\ref{thm:main2}, we will construct six instances, analyze the outputs of $\M$ on these instances, and prove that truthfulness and $(1-\tau)$-approximate proportionality cannot be both guaranteed.
The six instances we used are similar to those in the proof of Theorem~\ref{thm:main2} and are shown in Table~\ref{tab_append}.

\begin{table}
    \centering
    \begin{tabular}{|l|l|}
    \hline
    Instance & Allocation \\
    \hline
      \includegraphics{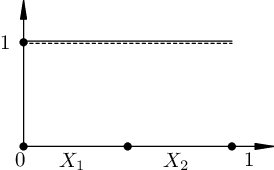}   &  $\M_1(F^{(1)})\subseteq X_1$ and $\M_2(F^{(1)})=X_2$\\
      \hline
      \includegraphics{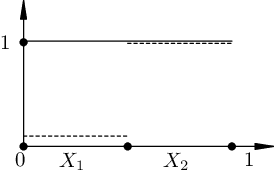}  & $\M_1(F^{(2)})\subseteq X_1$ and $\M_2(F^{(2)})=X_2$\\
      \hline
      \includegraphics{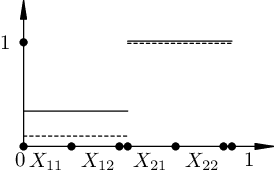}  & $\M(F^{(3)})=(X_{11}\cup X_{21},X_{12}\cup X_{22})$\\
      \hline
      \includegraphics{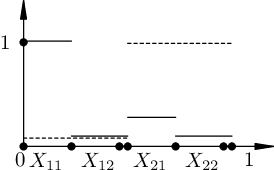}  & Proposition~\ref{prop:ins4_approx}\\
      \hline
      \includegraphics{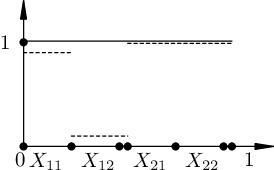}  &  Proposition~\ref{prop:ins5_approx}\\
      \hline
      \includegraphics{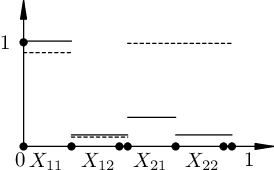}  &  
      \begin{tabular}{l}
           Proposition~\ref{prop:approx_final1} and Proposition~\ref{prop:approx_final2}  \\
           Later, we show that the two propositions cannot be held simultaneously.\\
           This yields a contradiction.
      \end{tabular}\\
      \hline
    \end{tabular}
    \caption{Instances constructed for the proof of Theorem~\ref{thm:main2_approx} and the corresponding allocations given by $\M$. The value density for agent $1$ is shown in solid lines, and the value density for agent $2$ is shown in dashed lines.}
    \label{tab_append}
\end{table}

\begin{instance}
$F^{(1)}=(f_1^{(1)},f_2^{(1)})$, where $f_1^{(1)}(x)=1$ and $f_2^{(1)}(x)=1$ for $x\in[0,1]$.
\end{instance}

To ensure the $(1-\tau)$-approximate proportionality, we must have $|\M_1(F^{(1)})|\geq\frac12(1-\tau)$ and $|\M_2(F^{(1)})|\geq\frac12(1-\tau)$.
Let $X_2=\M_2(F^{(1)})$ and $X_1=[0,1]\setminus X_2$.
We have $\M_1(F^{(1)})\subseteq X_1$.
Notice that we may have $\M_1(F^{(1)})\subsetneq X_1$, as we do not require $\M$ to be entire.

\begin{definition}
$X_2=\M_2(F^{(1)})$ and $X_1=[0,1]\setminus X_2$.
\end{definition}

Since $|\M_1(F^{(1)})|\geq\frac12(1-\tau)$ and $|\M_2(F^{(1)})|\geq\frac12(1-\tau)$, we have
\begin{equation}\label{eqn:X1X2}
    |X_1|,|X_2|\in\left[\frac12(1-\tau),\frac12(1+\tau)\right].
\end{equation}

In the instances constructed later, we let $\varepsilon>0$ be a sufficiently small real number.
Next, we consider the following instance.

\begin{instance}
$F^{(2)}=(f_1^{(2)},f_2^{(2)})$, where $f_1^{(2)}(x)=1$ for $x\in[0,1]$ and
$$f_2^{(2)}(x)=\left\{\begin{array}{ll}
    \varepsilon & x\in X_1 \\
    1 & x\in X_2
\end{array}\right..$$
\end{instance}

This instance is the same as the second instance in the proof of Theorem~\ref{thm:main2}, except that $X_1$ and $X_2$ are defined differently.

\begin{proposition}
$\M_1(F^{(2)})\subseteq X_1$ and $\M_2(F^{(2)})=X_2$.
\end{proposition}
\begin{proof}
Firstly, we must have $|\M_2(F^{(2)})|\leq|X_2|$.
Otherwise, in the first instance, agent $2$ will misreport $f_2^{(2)}$ instead of truthfully reporting $f_2^{(1)}$ and receive an interval with a length of more than $|X_2|$, which is more beneficial.
This will violate truthfulness.

Given $|\M_2(F^{(2)})|\leq|X_2|$, the maximum value agent $2$ can receive is $|X_2|$ by $\M_2(F^{(2)})=X_2$.
In addition, if agent $2$ reports $f_2^{(1)}$ instead of truthfully reporting $f_2^{(2)}$, the instance becomes $F^{(1)}$ and we know agent $2$ will receive $X_2$.
To guarantee truthfulness, we must have $\M_2(F^{(2)})=X_2$.

Finally, this further implies $\M_1(F^{(2)})\subseteq X_1$.
\end{proof}

The third instance is also similar to before.
To optimize the approximation ratio for proportionality in this impossibility result, we set the value for $f_1^{(3)}(x)$ on $X_1$ to $\frac13$ instead of $0.5$.

\begin{instance}
$F^{(3)}=(f_1^{(3)},f_2^{(3)})$, where
$$f_1^{(3)}(x)=\left\{\begin{array}{ll}
    \frac13 & x\in X_1 \\
    1 & x\in X_2
\end{array}\right.\qquad\mbox{and}\qquad
f_2^{(3)}(x)=\left\{\begin{array}{ll}
    \varepsilon & x\in X_1 \\
    1 & x\in X_2
\end{array}\right..$$
\end{instance}

We will define $X_{11},X_{12},X_{21}$ and $X_{22}$ as follows.

\begin{definition}
$X_{11}=\M_1(F^{(3)})\cap X_1$, $X_{12}=\M_2(F^{(3)})\cap X_1$, $X_{21}=\M_1(F^{(3)})\cap X_2$ and $X_{22}=\M_2(F^{(3)})\cap X_2$.
\end{definition}

We have $\M_1(F^{(3)})=X_{11}\cup X_{21}$ and $\M_2(F^{(3)})=X_{12}\cup X_{22}$.
We also have $|X_{11}|+|X_{12}|\leq |X_1|$ and $|X_{21}|+|X_{22}|\leq |X_2|$.
Notice that the inequalities may be strict, as the allocation needs not to be entire.

We show that both $|X_{11}|$ and $|X_{21}|$ are approximately $\frac14$.
The proof is similar to the proof of Proposition~\ref{prop:ins3}, with some extra calculations.

\begin{proposition}
$|X_{11}|$ and $|X_{21}|$ are bounded as follows:
$$\frac14-\frac72\tau+\frac14\tau^2-\varepsilon\cdot\frac34(1+\tau)^2\leq|X_{11}|\leq\frac14+\frac32\tau-\frac14\tau^2,$$
$$\frac14-\tau+\frac14\tau^2\leq|X_{21}|\leq\frac14(1+\tau)^2+\varepsilon\cdot\frac14(1+\tau)^2.$$
\end{proposition}
\begin{proof}
By the $(1-\tau)$-approximate proportionality for agent $1$, we must have
\begin{equation}\label{eqn:ins3agent1lowerbound}
\frac13|X_{11}|+|X_{21}|\geq\frac12(1-\tau)\cdot \left(\frac13|X_1|+|X_2|\right).
\end{equation}
In addition, we must also have $|\M_1(F^{(3)})|\leq|\M_1(F^{(2)})|$. Otherwise, in the second instance, it is more beneficial for agent $1$ to report $f_1^{(3)}$ than truthfully reporting $f_1^{(2)}$.
Thus,
\begin{equation}\label{eqn:ins3agent1upperbound}
|X_{11}|+|X_{21}|\leq |\M_1(F^{(2)})|\leq |X_1|.
\end{equation}
By (\ref{eqn:ins3agent1lowerbound}) and (\ref{eqn:ins3agent1upperbound}), we can obtain
\begin{equation}\label{eqn:lower_bound_X21}
    |X_{21}|\geq-\frac14(1+\tau)|X_1|+\frac34(1-\tau)|X_2|.
\end{equation}

By the $(1-\tau)$-approximate proportionality for agent $2$, we have
$$\varepsilon|X_{12}|+|X_{22}|\geq\frac12(1-\tau)\cdot\left(\varepsilon|X_1|+|X_2|\right),$$
which, by $|X_{12}|\leq |X_1|$, implies
$$|X_{22}|\geq\frac12(1-\tau)|X_2|+\varepsilon\cdot\left(\frac12(1-\tau)|X_1|-|X_{12}|\right)\geq \frac12(1-\tau)|X_2|-\varepsilon\cdot\frac12(1+\tau)|X_1|,$$
which, by $|X_{21}|+|X_{22}|\leq |X_2|$, further implies
\begin{equation}\label{eqn:upper_bound_X21}
|X_{21}|\leq \frac12(1+\tau)|X_2|+\varepsilon\cdot\frac12(1+\tau)|X_1|.
\end{equation}
Substituting (\ref{eqn:X1X2}) into (\ref{eqn:lower_bound_X21}) and (\ref{eqn:upper_bound_X21}), we have
\begin{equation}\label{eqn:X21}
    \frac14-\tau+\frac14\tau^2\leq|X_{21}|\leq\frac14(1+\tau)^2+\varepsilon\cdot\frac14(1+\tau)^2.
\end{equation}

We can also obtain the range of $|X_{11}|$ by combining (\ref{eqn:ins3agent1lowerbound}), (\ref{eqn:ins3agent1upperbound}), (\ref{eqn:X21}) and (\ref{eqn:X1X2}) with some calculations:
\begin{equation}\label{eqn:X11}
    \frac14-\frac72\tau+\frac14\tau^2-\varepsilon\cdot\frac34(1+\tau)^2\leq|X_{11}|\leq\frac14+\frac32\tau-\frac14\tau^2.
\end{equation}
\end{proof}

\begin{instance}
$F^{(4)}=(f_1^{(4)},f_2^{(4)})$, where
$$f_1^{(4)}(x)=\left\{\begin{array}{ll}
    1 & x\in X_{11} \\
    \sqrt{\varepsilon} & x\in X_{21}\\
    \varepsilon & x\in [0,1]\setminus(X_{11}\cup X_{21})
\end{array}\right.\qquad\mbox{and}\qquad
f_2^{(4)}(x)=\left\{\begin{array}{ll}
    \varepsilon & x\in X_1 \\
    1 & x\in X_2
\end{array}\right..$$
\end{instance}

The proposition below shows that the total length agent $2$ can get from $X_2$ is at most approximately $\frac14$.

\begin{proposition}\label{prop:ins4_approx}
$|\M_2(F^{(4)})\cap X_2|\leq\frac14+\frac32\tau-\frac14\tau^2+\sqrt\varepsilon$.
\end{proposition}
\begin{proof}
Suppose agent $1$ report $f_1^{(3)}$ instead of his/her true value density function $f_1^{(4)}$.
The instance becomes Instance 3, and we have seen that agent $1$ will receive $X_{11}\cup X_{21}$, which is worth $|X_{11}|+\sqrt\varepsilon|X_{21}|$ with respect to his/her true value density function $f_1^{(4)}$.
To ensure truthfulness, we must have 
\begin{equation}\label{eqn:ins4_valuation_lowerbound}
v_1(\M_1(F^{(4)}))\geq|X_{11}|+\sqrt\varepsilon\cdot|X_{21}|.
\end{equation}

On the other hand, we have
\begin{align*}
    v_1(\M_1(F^{(4)})) &= |\M_1(F^{(4)})\cap X_{11}|+\sqrt\varepsilon\cdot|\M_1(F^{(4)})\cap X_{21}|+\varepsilon\cdot|M_1(F^{(4)})\setminus(X_{11}\cup X_{21})|\\
    &\leq |X_{11}|+\sqrt\varepsilon\cdot|\M_1(F^{(4)})\cap X_{21}|+\varepsilon.
\end{align*}
Combining this with (\ref{eqn:ins4_valuation_lowerbound}), we have
$$|\M_1(F^{(4)})\cap X_{21}|\geq |X_{21}|-\sqrt\varepsilon.$$
For agent $2$, we then have
\begin{align*}
    |\M_2(F^{(4)})\cap X_2|&\leq |X_2|-|\M_1(F^{(4)})\cap X_2|\\
    &\leq |X_2|-|\M_1(F^{(4)})\cap X_{21}|\\
    &\leq |X_2|-|X_{21}|+\sqrt\varepsilon\\
    &\leq\frac12(1+\tau)-\left(\frac14-\tau+\frac14\tau^2\right)+\sqrt\varepsilon\tag{by (\ref{eqn:X1X2}) and (\ref{eqn:X21})}\\
    &=\frac14+\frac32\tau-\frac14\tau^2+\sqrt\varepsilon.
\end{align*}
\end{proof}

\begin{instance}
$F^{(5)}=(f_1^{(5)},f_2^{(5)})$, where $f_1^{(5)}(x)=1$ for $x\in[0,1]$ and
$$f_2^{(5)}(x)=\left\{\begin{array}{ll}
    1 & x\in X_2 \\
    1-\varepsilon & x\in X_{11}\\
    \varepsilon & x\in X_1\setminus X_{11}
\end{array}\right..$$
\end{instance}

The following proposition says that agent $1$ must receive most of $X_{11}$ and agent $2$ must receive exactly $X_2$.

\begin{proposition}\label{prop:ins5_approx}
$|\M_1(F^{(5)})\cap X_{11}|\geq |X_{11}|-\tau$ and $\M_2(F^{(5)})=X_2$.
\end{proposition}
\begin{proof}
The reason for $\M_2(F^{(5)})=X_2$ is similar as it is in the proof of Proposition~\ref{prop:ins5}:
firstly, we must have $|\M_2(F^{(5)})|\leq|X_2|$, for otherwise agent $2$ in Instance $1$ will misreport his/her value density function to $f_2^{(5)}$; secondly, given $|\M_2(F^{(5)})|\leq|X_2|$, the maximum value agent $2$ can get is $|X_2|$ by receiving $X_2$, and we must allocate $X_2$ to agent $2$ to avoid him/her to misreport $f_2^{(1)}$.
This proves the second half of the proposition.

Since $X_2$ is allocated to agent $2$ and $X_1=[0,1]\setminus X_2$, we have $\M_1(F^{(5)})\subseteq X_1$.
To guarantee $(1-\tau)$-approximate proportionality, we must have $|\M_1(F^{(5)})\cap X_1|\geq \frac12(1-\tau)$, which, by (\ref{eqn:X1X2}), implies
$$|X_1\setminus\M_1(F^{(5)})|=|X_1|-|M_1(F^{(5)})\cap X_1|\leq\frac12(1+\tau)-\frac12(1-\tau)=\tau.$$
As a result,
$$|X_{11}\setminus\M_1(F^{(5)})|\leq|X_1\setminus\M_1(F^{(5)})|\leq\tau,$$
which implies the first half of the proposition.
\end{proof}

\begin{instance}
$F^{(6)}=(f_1^{(6)},f_2^{(6)})$, where
$$f_1^{(6)}(x)=\left\{\begin{array}{ll}
    1 & x\in X_{11} \\
    \sqrt{\varepsilon} & x\in X_{21}\\
    \varepsilon & x\in [0,1]\setminus(X_{11}\cup X_{21})
\end{array}\right.\qquad\mbox{and}\qquad
f_2^{(6)}(x)=\left\{\begin{array}{ll}
    1 & x\in X_2 \\
    1-\varepsilon & x\in X_{11}\\
    \varepsilon & x\in X_1\setminus X_{11}
\end{array}\right..$$
\end{instance}

Firstly, the length agent $2$ receives on $X_2$ is at most approximately $\frac14$.
\begin{proposition}\label{prop:approx_final1}
$|\M_2(F^{(6)})\cap X_2|\leq\frac14+\frac32\tau-\frac14\tau^2+2\sqrt\varepsilon$.
\end{proposition}
\begin{proof}
Consider Instance 4 in this proof.
By Proposition~\ref{prop:ins4_approx}, the value agent $2$ can receive in $\M_2(F^{(4)})$, with respect to $f_2^{(4)}$, is at most
$$\varepsilon\cdot |\M_2(F^{(4)})\cap X_1|+1\cdot\left(\frac14+\frac32\tau-\frac14\tau^2+\sqrt\varepsilon\right)<\frac14+\frac32\tau-\frac14\tau^2+2\sqrt\varepsilon.$$
If $|\M_2(F^{(6)})\cap X_2|>\frac14+\frac32\tau-\frac14\tau^2+2\sqrt\varepsilon$, the subset $\M_2(F^{(6)})\cap X_2$ is worth more than $\frac14+\frac32\tau-\frac14\tau^2+2\sqrt\varepsilon$ with respect to $f_2^{(4)}$.
Then agent $2$ will report $f_2^{(6)}$ instead of the true value density function $f_2^{(4)}$ (now the instance becomes Instance 6 as $f_1^{(4)}=f_1^{(6)}$), and receive more benefit, which contradicts to the truthfulness. 
\end{proof}

Next, we show that most part of $|X_{11}|$ is not allocated to agent $2$.
\begin{proposition}\label{prop:approx_final2}
$|\M_2(F^{(6)})\cap X_{11}|\leq\tau+\sqrt\varepsilon$.
\end{proposition}
\begin{proof}
Suppose agent $1$ report $f_1^{(5)}$ instead of his/her true value density function $f_1^{(6)}$.
The instance becomes Instance 5, and Proposition~\ref{prop:ins5_approx} implies agent $1$ will receive a length of at least $|X_{11}|-\tau$ on $X_{11}$, which is worth $|X_{11}|-\tau$ with respect to $f_1^{(6)}$.
To guarantee truthfulness, we must have $v_1(\M_1(F^{(6)}))\geq|X_{11}|-\tau$.

On the other hand, we have
\begin{align*}
    v_1(\M_1(F^{(6)}))&=|\M_1(F^{(6)})\cap X_{11}|+\sqrt\varepsilon\cdot|\M_1(F^{(6)})\cap X_{21}|+\varepsilon\cdot|\M_1(F^{(6)})\setminus(X_{11}\cup X_{21})|\\
    &\leq|\M_1(F^{(6)})\cap X_{11}|+\sqrt\varepsilon.
\end{align*}
Putting those together, we have
$$|\M_1(F^{(6)})\cap X_{11}|+\sqrt\varepsilon\geq |X_{11}|-\tau,$$
which implies $|\M_1(F^{(6)})\cap X_{11}|\geq|X_{11}|-\tau-\sqrt\varepsilon$, which further implies $|\M_2(F^{(6)})\cap X_{11}|\leq\tau+\sqrt\varepsilon$.
\end{proof}

Finally, we show that Proposition~\ref{prop:approx_final1} and Proposition~\ref{prop:approx_final2} imply that the $(1-\tau)$-approximate proportionality cannot be satisfied for agent $2$ if $\tau$ is small.

The two propositions imply the following upper bound on the value agent $2$ gets:
\begin{align*}
    v_2(\M_2(F^{(6)}))&=|\M_2(F^{(6)})\cap X_{11}|\cdot(1-\varepsilon)+|\M_2(F^{(6)})\cap X_{2}|\cdot1+|\M_2(F^{(6)})\setminus(X_{11}\cup X_2)|\cdot\varepsilon\\
    &\leq|\M_2(F^{(6)})\cap X_{11}|+|\M_2(F^{(6)})\cap X_{2}|+\varepsilon\\
    &\leq (\tau+\sqrt\varepsilon)+\left(\frac14+\frac32\tau-\frac14\tau^2+2\sqrt\varepsilon\right)+\varepsilon\tag{Proposition~\ref{prop:approx_final1} and Proposition~\ref{prop:approx_final2}}\\
    &<\frac14+\frac52\tau-\frac14\tau^2+4\sqrt\varepsilon.
\end{align*}

On the other hand, we have
\begin{align*}
    v_2([0,1])&=|X_2|+(1-\varepsilon)|X_{11}|+\varepsilon\cdot|X_1\setminus X_{11}|\\
    &\geq|X_2|+(1-\varepsilon)|X_{11}|\\
    &\geq\frac12(1-\tau)+(1-\varepsilon)\left(\frac14-\frac72\tau+\frac14\tau^2-\varepsilon\cdot\frac34(1+\tau)^2\right)\tag{by (\ref{eqn:X1X2}) and (\ref{eqn:X11})}\\
    &>\frac34-4\tau+\frac14\tau^2-10\varepsilon,\tag{$10$ is a loose upper bound to the coefficient of $\varepsilon$}
\end{align*}
and the $(1-\tau)$-approximately proportional value for agent $2$ is
$$\frac12(1-\tau)v_2([0,1])>\frac12(1-\tau)\left(\frac34-4\tau+\frac14\tau^2-10\varepsilon\right)>\frac38-\frac{19}8\tau+\frac{17}8\tau^2-\frac18\tau^3-10\varepsilon.$$

Therefore, to guarantee the $(1-\tau)$-approximate proportionality for agent $2$, a necessary condition is
$$\frac14+\frac52\tau-\frac14\tau^2+4\sqrt\varepsilon>\frac38-\frac{19}8\tau+\frac{17}8\tau^2-\frac18\tau^3-10\varepsilon.$$

Elementary calculations show that
$$\frac14+\frac52\tau-\frac14\tau^2<\frac38-\frac{19}8\tau+\frac{17}8\tau^2-\frac18\tau^3$$
for $\tau\in[0,0.025969]$.
By considering a sufficiently small $\varepsilon$, the $(1-\tau)$-approximate proportionality cannot hold for agent $2$ if $\tau\leq0.025969$, which concludes Theorem~\ref{thm:main2_approx}.

\setcounter{instance}{0}

\section{Proof of Theorem~\ref{thm:main_restricted}}
\label{append:restricted}
The proof of this theorem is very similar to the proof of Theorem~\ref{thm:main2}.
We only need to modify the six instances slightly.

All the $\varepsilon$ terms in the six instances are replaced by $0$.
The number $0.5$ in Instance~\ref{ins3} is replaced by $1/k$.
The number $2\varepsilon$ in Instance~\ref{ins4} and \ref{ins6} is replaced by $1/k$.
The number $1-\varepsilon$ in Instance~\ref{ins5} and \ref{ins6} is replaced by $(k-1)/k$, which is larger than $1/k$ since $k\geq 3$.
The six instances are shown in Table~\ref{tab_restricted}.

\begin{table}
    \centering
    \begin{tabular}{|c|c|}
    \hline
    Instance & Allocation \\
    \hline
      \includegraphics{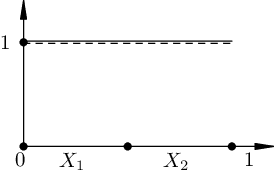}   & $\M(F^{(1)})=(X_1,X_2)$ \\
      \hline
      \includegraphics{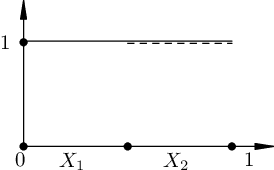}  & $\M(F^{(2)})=(X_1,X_2)$\\
      \hline
      \includegraphics{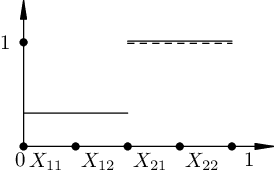}  & $\M_1(F^{(3)})=X_{11}\cup X_{21}$ and $X_{22}\subseteq\M_2(F^{(3)})\subseteq X_{12}\cup X_{22}$\\
      \hline
      \includegraphics{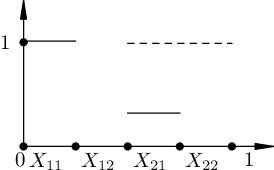}  & $\M_1(F^{(4)})=X_{11}\cup X_{21}$ and $X_{22}\subseteq\M_2(F^{(4)})\subseteq X_{12}\cup X_{22}$\\
      \hline
      \includegraphics{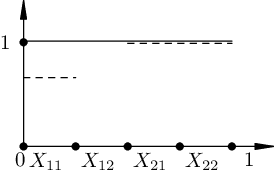}  & $\M(F^{(5)})=(X_1,X_2)$\\
      \hline
      \includegraphics{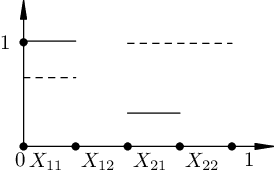}  & \begin{tabular}{l}
          Agent $1$ must receive a length of at least $\frac14-\frac1{4k}$ on $X_{11}$.  \\
          Agent $2$ must receive a length of at most $\frac14$ on $X_2=X_{21}\cup X_{22}.$ \\
          However, in this case, proportionality of agent $2$ fails.
      \end{tabular}\\
      \hline
    \end{tabular}
    \caption{Instances constructed for the proof of Theorem~\ref{thm:main_restricted} and the corresponding allocations given by $\M$. The value density for agent $1$ is shown in solid lines, and the value density for agent $2$ is shown in dashed lines.}
    \label{tab_restricted}
\end{table}

The analyses for the first two instances are exactly the same as before, and we have that $(X_1,X_2)$ is the only possible allocation.

For Instance~\ref{ins3}, we must have $|\M_1(F^{(3)})\cap X_1|=|\M_1(F^{(3)})\cap X_2|=|\M_2(F^{(3)})\cap X_2|=\frac14$.
The proof for $|\M_1(F^{(3)})\cap X_1|=|\M_1(F^{(3)})\cap X_2|=\frac14$ is the same as before, and $|\M_2(F^{(3)})\cap X_2|=\frac14$ is due to proportionality for agent $2$.
We let $X_{11}=\M_1(F^{(3)})\cap X_1$, $X_{12}=X_1\setminus X_{11}$,  $X_{21}=\M_1(F^{(3)})\cap X_2$, and $X_{22}=\M_2(F^{(3)})\cap X_2$.

For Instance~\ref{ins4}, we must also have $\M_1(F^{(4)})=X_{11}\cup X_{21}$ and $X_{22}\subseteq\M_2(F^{(4)})\subseteq X_{12}\cup X_{22}$.
The reason for $\M_1(F^{(4)})=X_{11}\cup X_{21}$ is the same as the proof of Proposition~\ref{prop:ins4}, and $X_{22}\subseteq\M_2(F^{(4)})$ is due to proportionality for agent $2$.

The analysis for Instance~\ref{ins5} is the same as the proof pf Proposition~\ref{ins5}, and we must have $\M(F^{(5)})=(X_1,X_2)$.

Now we derive a contradiction for Instance~\ref{ins6}.
Firstly, agent $1$ must receive a length of at least $\frac14-\frac1{4k}$ on $X_{11}$.
Otherwise, with respect to $f_1^{(6)}$, agent $1$ receives a value that is strictly less than $\frac14-\frac1{4k}$ on $X_{11}$, and agent $1$ receives a total value that is strictly less than $\frac14$ even if the entire $X_{21}$ is given to agent $1$.
However, by misreporting the value density function to $f_1^{(5)}$, agent $1$ gets $X_1$ which is worth $\frac14$ with respect to $f_1^{(6)}$.
Thus, the mechanism cannot be truthful.

Secondly, agent $2$ must receive a length of at most $\frac14$ on $X_2=X_{21}\cup X_{22}$.
Otherwise, in Instance~\ref{ins4}, reporting $f_2^{(6)}$ instead of reporting the true valuation $f_2^{(4)}$ is more beneficial for agent $2$.

However, the proportionality then fails to hold for agent $2$.
Agent $2$ can receive at most value $\frac{k-1}{4k^2}$ on $X_{11}$ and at most value $\frac14$ on $X_{22}$, which gives us the total value at most $\frac14+\frac{k-1}{4k^2}$, which is less than half of agent $2$'s value on the entire cake, $\frac12+\frac{k-1}{4k}$, given that $k\geq3$.

\section{Discussions on Brams, Jones, and Klamler's Truthful Notion and Mechanisms}
\label{append:Brams}
Brams, Jones, and Klamler~\cite{brams2006better} define a truthful notion called \emph{strategy-proofness} which is similar but slightly weaker than our risk-averse truthfulness.
In this section, we will use the word ``maximin strategy-proof'' to refer to the truthful notion defined by Brams, Jones, and Klamler (as strategy-proofness is more often used for dominant strategy truthfulness).
In Sect.~\ref{append:Brams1}, we will define maximin strategy-proofness and compare it with our (proportional) risk-averse truthfulness.
In Sect.~\ref{append:Brams2}, we will describe \emph{the equitability procedure}, a mechanism proposed by Brams, Jones, and Klamler that is maximin strategy-proof and proportional which always outputs allocations with connected pieces (see the first paragraph in Sect.~\ref{sect:RAT_P} for allocations with connected pieces), and we will compare it with our Mechanism~\ref{alg:rat_p}.

\subsection{Maximin Strategy-Proofness in Brams, Jones and Klamler~\cite{brams2006better}}
\label{append:Brams1}
A mechanism is \emph{maximin strategy-vulnerable} if a (risk-averse) agent can misreport his/her value density function and ``assuredly'' do better, regardless of the functions reported by other players.
A mechanism is \emph{maximin strategy-proof} if it is not maximin strategy-vulnerable.

\begin{definition}\label{def:maximinSP}
A mechanism $\M$ is \emph{maximin strategy-proof} if, for any agent $i$ with value density function $f_i$ and any value density function $f_i'$, there exist $f_1,\ldots,f_{i-1},f_{i+1},\ldots,f_n$ such that
$$v_i(\M_i(f_1,\ldots,f_{i-1},f_i,f_{i+1},\ldots,f_n))\geq v_i(\M_i(f_1,\ldots,f_{i-1},f_i',f_{i+1},\ldots,f_n)).$$
\end{definition}

It is clear from the definition that risk-averse truthfulness implies maximin strategy-proofness, as both 1 and 2 in Definition~\ref{def:wrat} imply the inequality in Definition~\ref{def:maximinSP}.
In fact, maximin strategy-proofness is slightly weaker than the risk-averse truthfulness (and so further weaker than the proportional risk-averse truthfulness).
Consider a scenario where an agent $i$ misreports $f_i$ to $f_i'$.
If $f_i$ and $f_i'$ give the same worst-case utility to agent $i$ and $f_i'$ sometimes performs strictly better, the mechanism is maximin strategy-proof, but it does not satisfy Definition~\ref{def:wrat}.

Maximin strategy-proofness is a rather weak notion if we consider the direct-revelation model (like in this paper) instead of Robertson-Webb query model.
Indeed, the following theorem says that every proportional mechanism automatically satisfies this truthful notion (a similar observation is proved by Ortega and Segal-Halevi; see Lemma~3 in reference~\cite{ortega2022obvious}.
Chen et al.~\cite{CL10} also call maximin strategy-proofness a “strikingly weak notion of truthfulness”

\begin{theorem}\label{thm:maximinSP_proportional}
All proportional mechanisms (including all those in Table~\ref{tab:mechanisms}) are maximin strategy-proof.
\end{theorem}
\begin{proof}
Consider an arbitrary proportional mechanism $\M$ and an arbitrary agent $i$ with value density function $f_i$.
Suppose the remaining $n-1$ agents report value density functions that are identical to $f_i$.
Consider an arbitrary function $f_i'$ reported by agent $i$, and let $(A_1,\ldots,A_n)$ be the allocation output by the mechanism.
By proportionality, we have $v_i(A_j)\geq\frac1nv_i([0,1])$ for each $j=1,\ldots,i-1,i+1,\ldots,n$ (since $f_i$ is the reported value density function for each of the remaining agents).
This implies $v_i(A_i)\leq\frac1nv_i([0,1])$.
Therefore, any report of agent $i$ would not yield an allocation that is worth more than agent $i$'s proportional value.
On the other hand, since $\M$ is proportional, agent $i$ will guarantee a proportional value when (s)he reports truthfully.
Therefore, $\M$ is maximin strategy-proof.
\end{proof}

\subsection{Equitability Procedure}
\label{append:Brams2}
Brams, Jones, and Klamler~\cite{brams2006better} propose a maximin strategy-proof mechanism, \emph{the equitability procedure}, that always outputs a proportional allocation.
In addition, Brams, Jones, and Klamler claim that, under the equitability procedure, an agent may receive a share that is worth less than his/her proportional share if (s)he misreports his/her value density function (see Theorem~3 of the paper).
This claim is even stronger than saying that the procedure is proportionally risk-averse truthful.
We will show that this claim is wrong, and the mechanism is not even proportionally risk-averse truthful.

\begin{definition}
Given a valuation profile $(f_1,\ldots,f_n)$, an allocation $(A_1,\ldots,A_n)$ is \emph{equitable} if
$$\frac{\int_{A_1}f_1(x)dx}{\int_0^1f_1(x)dx}=\frac{\int_{A_2}f_2(x)dx}{\int_0^1f_2(x)dx}=\cdots=\frac{\int_{A_n}f_n(x)dx}{\int_0^1f_n(x)dx}.$$
A mechanism is equitable if it always outputs equitable allocations with respect to the reported value density functions.
\end{definition}

The equitability procedure always outputs equitable, proportional, and entire allocations with connected pieces.
An entire allocation with connected pieces can be characterized by a permutation of $(1,\ldots,n)$ which specifies a left-to-right order of the agents and a set of $n-1$ cut points $x_1,\ldots,x_{n-1}$ that divide the cake to $n$ intervals such that the $i$-th interval is allocated to the $i$-th agent in the order specified by the permutation.
The equitability procedure computes $x_1,\ldots,x_{n-1}$ that yield an equitable allocation for each of the $n!$ permutations.
Then it outputs an allocation that maximizes the fractional value $\frac{\int_{A_i}f_i(x)dx}{\int_0^1f_i(x)dx}$ (notice that the fraction has the same value for all the agents, as the allocation is equitable).
This finishes the description of the equitability procedure.

It is proved in reference~\cite{brams2006better} that the procedure is maximin  strategy-proof.
In addition, the procedure always outputs a proportional allocation (stated in the first half of Theorem~3 in their paper).
Intuitively, in the moving-knife procedure, the first $n-1$ agents receive exactly their proportional shares, while the last agent may receive more.
Consider the same left-to-right order.
By shifting the $n-1$ cut points rightward for a little bit, we can make the allocation equitable while making sure each agent receives a piece with a slightly larger value.
Thus, there exist ``good'' left-to-right orders where the resultant equitable allocations are proportional.

As we mentioned, Brams, Jones, and Klamler misclaim in Theorem~3 that, under the equitability procedure, an agent may receive a piece with a value less than the proportional value if (s)he misreports his/her value density function.
Before we disprove this claim, we first note that value density functions are normalized with $\int_0^1f_i(x)dx=1$ in reference~\cite{brams2006better}, and two value density functions are considered the same if one rescales the other.
However, the uniform function $f(x)=1$ and the two functions $\ell^{(n)},r^{(n)}$ defined in (\ref{eqn:lr}) are all normalized, and they are distinct.
Suppose $f(x)=1$ is an agent's true value density function, Lemma~\ref{lem:lr} implies that reporting $\ell^{(n)}$ or $r^{(n)}$ can still guarantee a proportional share for this agent since the equitability procedure always outputs proportional allocations with respect to the reported value density functions.
Since $\ell^{(n)}$, $r^{(n)}$ and $f$ are different even up to normalization, this disproves the claim made by Brams, Jones, and Klamler.

In addition, the equitability procedure is not proportionally risk-averse truthful: an agent with the uniform value density function can misreport his/her valuation to $\ell^{(n)}$ or $r^{(n)}$, which is sometimes more beneficial while still guaranteeing to receive a proportional share.

\begin{theorem}\label{thm:equitability_PRAT}
The equitability procedure is not proportionally risk-averse truthful.
\end{theorem}
\begin{proof}
Let $f_1(x)=1$ for $x\in[0,1]$ be agent $1$'s true value density function.
Lemma~\ref{lem:lr} implies that reporting $\ell^{(n)}$ still guarantees a proportional share for agent $1$.
It remains to show that there exist $f_2,\ldots,f_n$ such that reporting $\ell^{(n)}$ is strictly more beneficial for agent $1$ than truthfully reporting $f_1$.
Let $f_2(x)=1$ for $x\in[0,\frac1n]$ and $f_2(x)=0$ for $x\in(\frac1n,1]$, and $f_3(x)=\cdots=f_n(x)=0$ for $x\in[0,\frac{n-1}n)$ and $f_3(x)=\cdots=f_n(x)=1$ for $x\in[\frac{n-1}n,1]$.

For both scenarios where agent $1$ reports $f_1$ and $\ell^{(n)}$ respectively, agent $1$ will be the second agent in the left-to-right order of the allocation.
Let $I$ be the interval allocated to agent $1$ when (s)he truthfully reports $f_1$.
Then $I$ is an interval that is near the left edge of the cake.
By misreporting $\ell^{(n)}$, the value of $I$ in terms of $\ell^{(n)}$ is smaller than its value in terms of $f_1$.
To maintain equitability, the equitability procedure will stretch $I$ to make sure the fractional value for agent $1$ matches the fractional value for the remaining agents.
This will make misreporting $\ell^{(n)}$ more beneficial.
\end{proof}

Notice that we do not know if the equitability procedure is  risk-averse truthful.

\paragraph{Comparison between the equitability procedure and Mechanism~\ref{alg:rat_p}, \ref{alg:rat_p'}}
The advantage of the equitability procedure is its equitability guarantee.
Equitability is a desirable property for fairness in many applications.
In our Mechanism~\ref{alg:rat_p}, the first agent in the left-to-right order receives exactly his/her proportional share, while the remaining agents may receive more than their proportional shares.
This may be viewed as being unfair to the first agent.

Mechanism~\ref{alg:rat_p'}, on the other hand, is equitable: each agent receives exactly $\frac1n$ of the value of the whole cake.
The equitability procedure outperforms this mechanism by allocation efficiency.
The equitability procedure always outputs entire allocations, and each agent may receive a piece with more than his/her proportional value.

The advantage of Mechanism~\ref{alg:rat_p} and \ref{alg:rat_p'} are their stronger truthful guarantees, as we have already seen.
In addition, the equitability procedure runs in exponential time (the mechanism needs to enumerate all the $n!$ permutations of the $n$ agents), while Mechanism~\ref{alg:rat_p}, as well as all our mechanisms in Sect.~\ref{sect:RAT_EF} and Sect.~\ref{sect:RAT_P}, run in polynomial time.

\section{Discussions on Troyan and Morrill's Obvious Manipulations}
\label{append:nom}
In this section, we formally introduce Troyan and Morrill's truthful notion of \emph{not obvious manipulability}, and compare it with our (proportional) risk-averse truthfulness.
Similar to our setting of (proportional) risk-averse truthfulness, Troyan and Morrill~\cite{troyan2020obvious} take into account that an agent does not know the utility functions of other agents.
A misreport of utility function is \emph{profitable} if there exists a set of other agents' utility function profile such that the misreporting agent is strictly better off. 
A profitable misreport of utility function is considered an \emph{obvious manipulation} if it either makes the agent strictly better off in the best case or makes the agent strictly better off in the worst case.
Below, we define not obvious manipulability in the context of cake cutting.

\begin{definition}
Given a mechanism $\M$, for an agent $i$ with value density function $f_i$, a value density function $f_i'$ is a \emph{profitable manipulation} if there exist $f_1,\ldots,f_{i-1},f_{i+1},\ldots,f_n$ such that
    $$v_i(\M_i(f_1,\ldots,f_{i-1},f_i,f_{i+1},\ldots,f_n))<v_i(\M_i(f_1,\ldots,f_{i-1},f_i',f_{i+1},\ldots,f_n)).$$
\end{definition}

\begin{definition}\label{def:nom}
A mechanism $\M$ is \emph{not obvious manipulable} if, for each agent $i$ with value density function $f_i$ and for any profitable manipulation $f_i'$, the following are true:
\begin{enumerate}
    \item $\displaystyle\inf_{f_1,\ldots,f_{i-1},f_{i+1},\ldots,f_n}v_i(\M_i(f_1,\ldots,f_{i-1},f_i,f_{i+1},\ldots,f_n))\geq\inf_{f_1,\ldots,f_{i-1},f_{i+1},\ldots,f_n}v_i(\M_i(f_1,\ldots,f_{i-1},f_i',f_{i+1},\ldots,f_n))$
    \item $\displaystyle\sup_{f_1,\ldots,f_{i-1},f_{i+1},\ldots,f_n}v_i(\M_i(f_1,\ldots,f_{i-1},f_i,f_{i+1},\ldots,f_n))\geq\sup_{f_1,\ldots,f_{i-1},f_{i+1},\ldots,f_n}v_i(\M_i(f_1,\ldots,f_{i-1},f_i',f_{i+1},\ldots,f_n))$
\end{enumerate}
\end{definition}

Besides the obvious difference that the not obvious manipulability focuses exclusively on the worse case and the best case, there is also a technical difference between it and (proportional) risk-averse truthfulness. 
In the definition above, a separate min (max) function is taken over the other agents’ strategies $f_1,\ldots,f_{i-1},f_{i+1},\ldots,f_n$ on both sides of the inequality. 
However, in the notion of risk-averse truthfulness, we are essentially considering if
$$\inf_{f_1,\ldots,f_{i-1},f_{i+1},\ldots,f_n}\left(u_i(\M_i(f_1,\ldots,f_{i-1},f_i',f_{i+1},\ldots,f_n))-u_i(\M_i(f_1,\ldots,f_{i-1},f_i,f_{i+1},\ldots,f_n))\right)$$
is negative. 
In particular, a single set of strategies $\{f_1,\ldots,f_{i-1},f_{i+1},\ldots,f_n\}$ is considered and taken to minimize the utility gain.

The strength of not obvious manipulability is incomparable with (proportional) risk-averse truthfulness.
As we will show later, our Mechanism~\ref{alg:rat_ef} is obviously manipulable, but we have seen that it is proportionally risk-averse truthful (Theorem~\ref{thm:rat_ef}).
On the other hand, Dubins-Spanier's moving-knife procedure is not obviously manipulable\footnote{Ortega and Segal-Halevi~\cite{ortega2022obvious} prove that the moving-knife procedure is obviously manipulable. However, their result applies to the Robertson-Webb query model, whereas we consider direct-revelation setting in our paper. See our remark following Theorem~\ref{thm:moving_knife_nom}.}, but we have seen that it is not risk-averse truthful (Theorem~\ref{thm:moving-knife}).

\begin{theorem}\label{thm:rat_ef_nom}
Mechanism~\ref{alg:rat_ef} is obviously manipulable.
\end{theorem}
\begin{proof}
Suppose agent $1$'s value density function is
$$f_1(x)=\left\{\begin{array}{ll}
    1 & x\in[0,\frac1n) \\
    0 & x\in[\frac1n,1]
\end{array}\right..$$
In the best case, if agent $1$ reports $f_1$ truthfully, regardless of the reports of the other $n-1$ agents, (s)he will receive a $\frac1n$ fraction of the interval $[0,\frac1n)$, which has value $\frac1{n^2}$.
However, when agent $1$ report the uniform function $f_1(x)=1$, (s)he will receive almost all of the interval $[0,\frac1n)$ if the other $n-1$ agents' value density functions are 
$$f_2(x)=\cdots=f_n(x)=\left\{\begin{array}{ll}
    0 & x\in[0,1-\varepsilon) \\
    1 & x\in[1-\varepsilon,1]
\end{array}\right.,$$
where $\varepsilon$ is a sufficiently small number.
Therefore, 2 in Definition~\ref{def:nom} fails.
\end{proof}

\begin{theorem}\label{thm:moving_knife_nom}
Dubins-Spanier's moving-knife procedure is not obviously manipulable.
\end{theorem}
\begin{proof}
Consider an arbitrary agent $i$ with value density function $f_i$ and any value density function $f_i'$.
In the worst case, agent $i$ will receive a value of $\frac1nv_i([0,1])$ when truthfully reporting $f_i$ (as long as (s)he is not the last agent who is allocated in the moving-knife procedure).
When reporting $f_i'$, agent $i$ will receive a value of at most $\frac1nv_i([0,1])$ in the worst case.
This happens when the remaining $n-1$ agents' value density functions are identical to $f_i$, in which case each of the remaining $n-1$ agents receives a piece with value at least $\frac1nv_i([0,1])$.
Thus, 1 in Definition~\ref{def:nom} holds.

In the best case, agent $i$ will receive almost the entire cake, which happens when the other $n-1$ agents' value density functions are 
$$f(x)=\left\{\begin{array}{ll}
    1 & x\in[0,\varepsilon) \\
    0 & x\in[\varepsilon,1]
\end{array}\right.,$$
where $\varepsilon$ is a sufficiently small number satisfying $v_i([0,\varepsilon))<\frac1nv_i([0,1])$.
By taking $\varepsilon\rightarrow0$, we can see the left-hand side of 2 in Definition~\ref{def:nom} reaches the maximum possible value $v_i([0,1])$.
Thus, 2 in Definition~\ref{def:nom} holds.
\end{proof}

As a remark, Ortega and Segal-Halevi~\cite{ortega2022obvious} prove that the moving-knife procedure is obviously manipulable if we are considering the Robertson-Webb query model where the game is formulated as an extensive-form game.
On the other hand, our claim above applies to the direct-revelation setting where all the $n$ value density functions are reported at the same time.
We refer the readers to Theorem~1 and Remark~1 of Ortega and Segal-Halevi's paper for the reason why the moving-knife procedure fails to be not obviously manipulable in the Robertson-Webb setting.

\subsection{Ortega and Segal-Halevi's Moving-Knife Procedure}
\label{append:OSmovingknife}
As we mentioned earlier, Ortega and Segal-Halevi~\cite{ortega2022obvious} prove that the moving-knife procedure is obviously manipulable under the Robertson-Webb query model.
As a solution, they propose a variant of the moving-knife procedure that is not obviously manipulable under the Robertson-Webb query model.

\paragraph{Ortega and Segal-Halevi's moving-knife procedure}
Same as Dubins-Spanier's moving-knife procedure, Ortega and Segal-Halevi's moving-knife procedure asks each agent $i$ a point $x_i$ such that $[0,x_i)$ is worth exactly the proportional value $\frac1nv_i([0,1])$, and allocate $[0,x_{i^\ast})$ to agent $i^\ast$ with the smallest $x_i$ value.
The difference between the two moving-knife procedures comes at the next step.
At the next step, Ortega and Segal-Halevi's moving-knife procedure asks each of the remaining $n-1$ agents for a point $x_i'$ such that $[x_{i^\ast},x_i')$ is worth exactly $\frac1{n-1}v_i([x_{i^\ast},1])$ (instead of the proportional value $\frac1nv_i([0,1])$ as it is in Dubins-Spanier's moving-knife procedure).
Similarly, at Step~$t$ when the remaining part of the cake is allocated among the remaining $n-t+1$ agents, each agent is asked to mark at a point such that the interval from the left endpoint to this point is worth exactly $\frac1{n-t+1}$ fraction of the value of the remaining part of the cake.
This is repeated until the $(n-1)$-th agent is allocated an interval, and then the last agent gets the remaining part of the cake.
It is easy to see that this variant of the moving-knife procedure always produces proportional allocations.
Unlike Dubins-Spanier's moving-knife procedure, all the agents, except for the first agent, may receive allocations with more than their proportional values.

\bigskip

Since Ortega and Segal-Halevi's moving-knife procedure is not obviously manipulable under the Robertson-Webb query model, it is also not obviously manipulable under our direct-revelation model.\footnote{A truthful property satisfied at every internal node in the tree of an extensive-form game is naturally satisfied at the root node.}
However, it does not satisfy our proportional risk-averse truthfulness.

\begin{theorem}\label{thm:OS_PRAT}
Ortega and Segal-Halevi's moving-knife procedure is not proportionally risk-averse truthful.
\end{theorem}
\begin{proof}
Let $f_1(x)=1$ be the true value density function for agent~$1$.
We show that agent~$1$ can misreport his/her value density function to $f_1'=\ell^{(n)}$ (see Eqn.~(\ref{eqn:lr})) that satisfies 1) there exist $f_2,\ldots,f_n$ such that $v_1(\M_1(f_1',f_2,\ldots,f_n))>v_1(\M_1(f_1,f_2,\ldots,f_n))$, and 2) for any $f_2,\ldots,f_n$, $v_1(\M_1(f_1',f_2,\ldots,f_n))\geq \frac1nv_1([0,1])$.

To see 1), suppose $f_2(x)=1$ for $x\in[0,\frac1n]$ and $f_2(x)=0$ for $x\in(\frac1n,1]$, and $f_3(x)=\cdots=f_n(x)=0$ for $x\in[0,\frac1n)$ and $f_3(x)=\cdots=f_n(x)=1$ for $x\in[\frac1n,1]$.
In Ortega and Segal-Halevi's moving-knife procedure, if agent $1$ truthfully reports $f_1$, (s)he will be the second agent receiving an interval after agent $2$ taking $[0,\frac1{n^2})$, and (s)he will receive $[\frac1{n^2},\frac1{n^2}+\frac1{n-1}\cdot(1-\frac1{n^2}))$, which is worth $\frac{n+1}{n^2}$.
If agent $1$ reports $f_1'$, (s)he will also be the second agent receiving an interval after agent $2$ taking $[0,\frac1{n^2})$, and (s)he will receive $[\frac1{n^2},\frac1n+\frac{5n-6}{2n^2(n-1)})$ (by some simple calculations), which is worth more than $\frac{n+1}{n^2}$ with respect to his/her true valuation.

To see 2), suppose agent $1$ reports $f_1'$.
Since Ortega and Segal-Halevi's moving-knife procedure is proportional, regardless of what the remaining $n-1$ agents report, agent $1$ will receive an interval that has a value of at least $\frac1n$ with respect to $f_1'$.
By Lemma~\ref{lem:lr}, agent $1$ receives an interval that is worth at least $\frac1n$ with respect to his/her true valuation $f_1$.
\end{proof}

The risk-averse truthfulness of Ortega and Segal-Halevi's moving-knife procedure depends on a subtle tie-breaking issue.
We discuss it in \ref{append:OSRAT}.

\section{The missing proofs for Table~\ref{tab:mechanisms}}
\label{append:missingproofs}
In this section, we provide the proofs for the results in Table~\ref{tab:mechanisms} that are missing in the previous parts of the paper.

\subsection{Proofs for the Satisfability of Not Obvious Manipulability}
Here, we prove the results in the last column of Table~\ref{tab:mechanisms}.
We first present the following lemma whose proof is almost the same as the proof of Theorem~\ref{thm:maximinSP_proportional} and is also available in Reference~\cite{ortega2022obvious} (Lemma~3 in the paper).

\begin{lemma}\label{lem:1nom}
Every proportional mechanism satisfies 1 in Definition~\ref{def:nom}.
\end{lemma}

The proofs for all the three theorems below are almost identical.
We will show them in one proof.

\begin{theorem}\label{thm:Even-Paz_nom}
The Even-Paz algorithm is not obviously manipulable.
\end{theorem}

\begin{theorem}\label{thm:equitability_nom}
The equitability procedure is not obviously manipulable.
\end{theorem}

\begin{theorem}\label{thm:rat_p_nom}
Mechanism~\ref{alg:rat_p} is not obviously manipulable.
\end{theorem}

\begin{proof}[Proof of Theorem~\ref{thm:Even-Paz_nom}, \ref{thm:equitability_nom}, and \ref{thm:rat_p_nom}]
Lemma~\ref{lem:1nom} implies 1 in Definition~\ref{def:nom} holds.
To show 2 in the definition, it suffices to show that a truthful agent can get almost the entire cake in the best case.
In all the three mechanisms (Even-Paz, equitability, Mechanism~\ref{alg:rat_p}), this happens when all the remaining agents only have a positive value on the interval $[0,\varepsilon)$ for an arbitrarily small $\varepsilon>0$.
\end{proof}

Finally, Mechanism~\ref{alg:rat_p'} is obviously manipulable.

\begin{theorem}\label{thm:rat_p'_nom}
Mechanism~\ref{alg:rat_p'} is obviously manipulable.
\end{theorem}
\begin{proof}
Suppose agent $1$'s value density function is $f_1(x)=1$ on $[0,1]$.
(S)he will always receive a length of $\frac1n$ when truth-telling, which is worth $\frac1n$, regardless of the reports of the remaining $n-1$ agents.
However, if (s)he reports
$$f_1'(x)=\left\{\begin{array}{ll}
    \frac{n-1}\varepsilon &  x\in[0,\varepsilon)\\
    1 & x\in[\varepsilon,1]
\end{array}\right.$$
for some small $\varepsilon>0$, there is a chance (s)he receives $[\varepsilon,1]$ (when the remaining agents only have positive values on $[0,\varepsilon)$), which, in terms of the true value density function $f_1$, is worth much more than $\frac1n$.
\end{proof}

\subsection{On Risk-Averse Truthfulness of Even-Paz Algorithm}
\label{append:EvenPaz}
In this section, we will show that, to make the Even-Paz algorithm risk-averse truthful, we need to carefully choose a tie-breaking rule.

Recall that the algorithm is based on divide-and-conquer.
At each recursive call where an interval $I\subseteq[0,1]$ is allocated to a subset of $k$ agents, each agent is asked to report his/her $\lfloor\frac k2\rfloor:\lceil\frac k2\rceil$ point, $I$ is then cut at the median of these reported points, and the two halves of $I$ are allocated recursively.
It is possible that the $\lfloor\frac k2\rfloor:\lceil\frac k2\rceil$ point is not unique for some agent, which happens when this agent's value density function is $0$ on some intervals (e.g., when $f(x)=1$ on $[0,0.4)\cup[0.6,1]$ and $f(x)=0$ on $[0.4,0.6)$, the $0.5:0.5$ point can be any point on the interval $[0.4,0.6]$).
Therefore, to complete the description of the algorithm, we need to specify a tie-breaking rule.
We consider two natural tie-breaking rules: always choose the left-most point ($0.4$ in the example above) and always choose the right-most point ($0.6$ in the example above).

We have described the Even-Paz algorithm such that the \emph{median} of the $k$ points is cut at each recursive call.
Many papers define the algorithm such that the $\lfloor\frac k2\rfloor$-th point is cut.
We will also consider this variant in this section.

Thus, we consider the following four different implementations of the Even-Paz algorithm.
Below, an agent's ``middle point'' refers to the $\lfloor\frac k2\rfloor:\lceil\frac k2\rceil$ point of this agent when a subinterval $I$ is allocated to $k$ agents at a recursive call.
\begin{itemize}
    \item Even-Paz(median, left): The median of the $k$ middle points is cut, and the left-most point is used when an agent's middle point is not unique.
    \item Even-Paz(median, right): The median of the $k$ middle points is cut, and the right-most point is used when an agent's middle point is not unique.
    \item Even-Paz($\lfloor\frac k2\rfloor$, left): The $\lfloor\frac k2\rfloor$-th middle point is cut, and the left-most point is used when an agent's middle point is not unique.
    \item Even-Paz($\lfloor\frac k2\rfloor$, right): The $\lfloor\frac k2\rfloor$-th middle point is cut, and the right-most point is used when an agent's middle point is not unique.
\end{itemize}

It is easy to verify that all the implementations make the algorithm output proportional allocations.
We will show that Even-Paz($\lfloor\frac k2\rfloor$, left) is not risk-averse truthful, and the remaining three implementations are risk-averse truthful.

\begin{theorem}\label{thm:Even-Paz_RAT}
Even-Paz(median, left), Even-Paz(median, right) and Even-Paz($\lfloor\frac k2\rfloor$, right) are risk-averse truthful. Even-Paz($\lfloor\frac k2\rfloor$, left) is not risk-averse truthful.
\end{theorem}

Notice that the left-variant and the right-variant are identical if we assume agents are hungry.
\begin{corollary}
The Even-Paz algorithm is risk-averse truthful if agents are hungry.
\end{corollary}

In the remaining part of this section, we prove Theorem~\ref{thm:Even-Paz_RAT}.

\subsubsection{Even-Paz(median, left), Even-Paz(median, right) and Even-Paz($\lfloor\frac k2\rfloor$, right) are risk-averse truthful}
Without loss of generality, we consider agent $1$ with value density function $f_1$ who reports $f_1'$ instead.
Consider the executions of $\M(f_1,f_2,\ldots,f_n)$ and $\M(f_1',f_2,\ldots,f_n)$, where $\M$ stands for any one of the three implementations of the Even-Paz algorithm.
If agent $1$'s middle point is the same at each recursive call throughout the algorithm for any $f_2,\ldots,f_n$, then reporting $f_1'$ is non-beneficial, and 2(a) of Definition~\ref{def:rat} holds.
Therefore, we assume, for some $f_2,\ldots,f_n$, at some recursive call where a subinterval $I=[s,t]$ is allocated to a subset of $k\geq2$ agents, agent $1$'s middle points are different.
Let $x_1$ and $x_1'$ be agent $1$'s middle points on $I$ with respect to $f_1$ and $f_1'$ respectively.
Let $x$ and $x'$ be the corresponding cut points for this recursive call (which is the median of all the middle points in Even-Paz(median, left) and Even-Paz(median, right), and the $\lfloor\frac k2\rfloor$-th middle point in Even-Paz($\lfloor\frac k2\rfloor$, right)).

First of all, we show the following proposition.
\begin{proposition}
For any of the three implementations of the Even-Paz algorithm, if $v_1([s,x_1])\neq v_1([s,x_1'])$, there is always a chance that agent $1$ receives less value by reporting $f_1'$ instead of $f_1$.
\end{proposition}
\begin{proof}
We will only prove the proposition for the case $v_1([s,x_1])<v_1([s,x_1'])$, as the other case $v_1([s,x_1])>v_1([s,x_1'])$ is similar.
In this case, $v_1([x_1,t])>v_1([x_1',t])$.
By the nature of the Even-Paz algorithm, if agent $1$ reports $f_1$ truthfully, (s)he can guarantee receiving an interval with value at least $\frac1kv_1([s,t])$.
It suffices to show that agent $1$ has a chance to get an interval with value strictly less than $\frac1kv_1([s,t])$ by reporting $f_1'$.

In all the three implementations, for any $\varepsilon>0$, it is possible that, by reporting $f_1'$, we have $x'\in(x_1'-\varepsilon,x_1']$ and agent $1$ is among those $\lceil\frac k2\rceil$ agents to whom an allocation of $[x',t]$ is decided in the next recursive call.
In this case, agent $1$'s value on $[x',t]$ (which is approximately the value on $[x_1',t]$) is strictly less than the value on $[x_1,t]$ (for sufficiently small $\varepsilon$), and thus strictly less than $\frac{\lceil\frac k2\rceil}{k}v_1([s,t])$.
In addition, it is possible that agent $1$ will end up with receiving at most a $\frac1{\lceil \frac k2\rceil}$ fraction of the value of $[x',t]$ (which happens when the remaining agents' value density functions are identical, up to rescaling, to $f_1$ on the interval $[x',t]$).
This value is less than $\frac1kv_1([s,t])$.
\end{proof}

From now on, we assume $v_1([s,x_1])=v_1([s,x_1'])$.
Next, we prove the risk-averse truthfulness of the three implementations by consider different cases for $k$.

\paragraph{Case $k\geq4$}
In this case, both $\lfloor \frac k2\rfloor$ and $\lceil\frac k2\rceil$ are at least $2$.
Suppose $x_1'>x_1$.
It is possible that all the remaining agents' middle points are between $x_1$ and $x_1'$,
we have $x_1'> x'=x> x_1$.
Since $v_1([s,x_1])=v_1([s,x_1'])$, we have $v_1([s,x])=\frac{\lfloor\frac k2\rfloor}kv_1([s,t])$ and $v_1([x,t])=\frac{\lceil\frac k2\rceil}kv_1([s,t])$.
Moreover, agent $1$ is among those $\lfloor \frac k2\rfloor$ agents to whom an allocation of $[s,x]$ is decided at the next recursive call if agent $1$ is truth-telling, and agent $1$ is among those $\lceil \frac k2\rceil$ agents to whom an allocation of $[x,t]$ is decided at the next recursive call if agent $1$ reports $f_1'$.
If $f_2,\ldots,f_n$ are positive for only a very tiny interval on $[s,x]$ and are identical (up to rescaling) with $f_1$ on $[x,t]$, then, at the end of the algorithm, agent $1$ gets a value that is very close to $\frac{\lfloor \frac k2\rfloor}k v_1([s,t])$ when truth-telling and gets a value of at most $\frac1kv_1([s,t])$ when reporting $f_1'$.
Since $\lfloor \frac k2\rfloor\geq 2$, this shows that misreporting can be potentially harmful for agent $1$.
The analysis for the case $x_1'<x_1$ is similar.

\paragraph{Case $k=3$}
We first consider the two implementations Even-Paz(median, right) and Even-Paz($\lfloor \frac k2\rfloor$, right).
Since we have assumed $v_1([s,x_1])=v_1([s,x_1'])$ and both implementations break the tie by choosing the right-most middle point, we must have $x_1'<x_1$.
It is possible that all the remaining two agents' middle points are between $x_1'$ and $x_1$,
we have $x_1> x'\geq x\geq x_1'$.
Since $v_1([s,x_1])=v_1([s,x_1'])$, we have $v_1([s,x])=\frac13v_1([s,t])$ and $v_1([x,t])=\frac23v_1([s,t])$.
Moreover, agent $1$ is among those $2$ agents to whom an allocation of $[x,t]$ is decided at the next recursive call if agent $1$ is truth-telling, and agent $1$ will receive $[s,x]$ if agent $1$ reports $f_1'$.
If $f_2$ and $f_3$ are positive for only a very tiny interval on $[x,t]$, then agent $1$ gets a value that is very close to $\frac23 v_1([s,t])$ when truth-telling and gets a value of exactly $\frac13v_1([s,t])$ when reporting $f_1'$.
This shows that misreporting can be potentially harmful for agent $1$.

We then consider the implementation Even-Paz(median, left).
In this case, we must have $x_1'>x_1$.
It is possible that $f_2$ and $f_3$ satisfy the following.
\begin{itemize}
    \item agent $2$'s middle point ($\frac13:\frac23$ point) is to the left of $x_1$.
    \item agent $3$'s middle point ($\frac13:\frac23$ point) is to the right of $x_1'$.
    \item agent $3$ has a positive value on $[x_1,x_1']$.
    \item $f_3$ is identical to $f_1$ on $[x_1',t]$ up to rescaling.
\end{itemize}
In this case, $x=x_1$ and $x'=x_1'$.
Consider the last round where agent $1$ and $3$ will be allocated from $[x,t]$ (when agent $1$ reports $f_1$) or $[x',t]$ (when agent $1$ reports $f_1'$).
If agent $1$ reports $f_1'$, agent $1$ will receive at most half of the value of $[x_1',t]$ at the end.
If agent $1$ truthfully reports $f_1$, since agent $3$ has a positive value on $[x_1,x_1']$ and $f_3$ and $f_1$ are identical on $[x_1',t]$, agent $3$'s middle point does not split $[x_1',t]$ evenly in agent $1$'s valuation, and agent $1$ will receive a piece that is worth more than half of the value of $[x_1',t]$.
Agent $1$'s misreporting is potentially harmful.

\paragraph{Case $k=2$}
For Even-Paz(median, left), we must have $x_1'>x_1$ since $v_1([s,x_1])=v_1([s,x_1'])$.
It is easy to verify that, when the other agent's middle point is in $[x_1-\delta,x_1]$ where $f_1$ is positive on $[x_1-\delta,x_1]$ (for some $\delta>0$), misreporting is harmful for agent $1$.

For Even-Paz(median, right), we must have $x_1'<x_1$ since $v_1([s,x_1])=v_1([s,x_1'])$.
It is easy to verify that, when the other agent's middle point is in $[x_1,x_1+\delta]$ where $f_1$ is positive on $[x_1,x_1+\delta]$ (for some $\delta>0$), misreporting is harmful for agent $1$.

For Even-Paz($\lfloor \frac k2\rfloor$, right), we must have $x_1'<x_1$ since $v_1([s,x_1])=v_1([s,x_1'])$.
It is easy to verify that misreporting is always non-beneficial (and also non-harmful) in all the three cases, $x_2\in [s,x_1)$, $x_2\in [x_1,x_1')$ and $x_2\in[x_1',t]$, where $x_2$ is the middle point for the other agent.

\subsubsection{Even-Paz($\lfloor \frac k2\rfloor$, left) is not risk-averse truthful}
\label{append:Even-Paz-not-RAT}
We consider an instance with three agents.
We describe the instance with the cake being represented by $[0,5]$, which can be scaled to $[0,1]$ with the value density functions scaled accordingly.
Consider
\begin{equation}\label{eqn:EvenPazRAT}
f(x)=\left\{\begin{array}{ll}
    1 & x\in[0,1)\cup[2,3)\cup[4,5] \\
    0 & x\in[1,2)\cup[3,4)
\end{array}\right.\qquad\mbox{and}\qquad f'(x)=\left\{\begin{array}{ll}
    1 & x\in[0,1)\cup[2-\varepsilon,3)\cup[4-\varepsilon,5] \\
    0 & x\in[1,2-\varepsilon)\cup[3,4-\varepsilon)
\end{array}\right.
\end{equation}
for sufficiently small $\varepsilon>0$.
Suppose $f$ is the true value density function for agent $1$. We will show that reporting $f'$ is sometimes beneficial and always not harmful.

To see it is sometimes beneficial, suppose agent $2$ and $3$ only have positive values on $[1+\varepsilon,2-\varepsilon]$.
When truth-telling, agent $1$'s middle point is at $x_1=1$, which is to the left of the other two agents' middle points.
(S)he will receive $[0,1)$, which is worth value $1$.
On the other hand, by reporting $f'$, agent $1$'s middle point is between $2-\varepsilon$ and $2$, which is to the right of the other two agents' middle points.
In addition, at the final round, the other agent's middle point is still in $[1+\varepsilon,2-\varepsilon]$, and agent $1$ will receive a value of $2$. 

To see reporting $f'$ is always not harmful, we discuss the following two cases.
\paragraph{Case 1}
Suppose one of agent $2$ or $3$ has its middle point in $[0,1)$.
Assume without loss of generality that $x_2\in[0,1)$.
In this case agent $2$ gets $[0,x_2)$, and $[x_2,1]$ is allocated to agent $1$ and $3$ at the final round.
Simple calculations reveal that the middle points of both $f$ and $f'$ on $[x_2,1]$ are identical.
This means that misreporting will not change the output allocation at all.

\paragraph{Case 2}
Suppose no agent has its middle point in $[0,1)$.
When agent $1$ truthfully reports $f$, (s)he will receive $[0,1)$ which is worth $1$.
It suffices to show that agent $1$ will receive a value of at least $1$ when reporting $f'$.
Let $x_1'=2-\frac13\varepsilon$ be agent $1$'s middle point with respect to $f_1'$.
Without loss of generality, suppose the other two agents' middle points satisfy $x_2\leq x_3$.
We further discuss three cases regarding $x_2$.

If $x_2\geq x_1'$, agent $1$'s middle point is still the leftmost among the three agents, and (s)he will receive $[0,x_1')$ which is worth exactly $1$.

If $x_2\in(2-\varepsilon,x_1')$, agent $2$ will get $[0,x_2)$, and agent $1$ and $3$ is allocated the interval $[x_2,1]$ at the final round.
It is easy to see that, at the final round, agent $1$'s middle point (with respect to $f'$) is between $4-\varepsilon$ and $4$.
In terms of agent $1$'s true valuation, this middle point splits $[x_2,1]$ into two intervals with the equal value $1$.
Agent $1$ will receive a value of at least $1$.

If $x_2\in[1,2-\varepsilon]$, agent $2$ will get $[0,x_2)$, and agent $1$ and $3$ is allocated the interval $[x_2,1]$ at the final round.
For both $f$ and $f'$, the middle point of agent $1$ at the final round is at $3$.
Reporting $f'$ does not change the output allocation in this case.

\subsection{On Risk-Averse Truthfulness of Ortega and Segal-Halevi's Moving-Knife Procedure}
\label{append:OSRAT}
Similar to the Even-Paz algorithm, the risk-averse truthfulness of Ortega and Segal-Halevi's moving-knife procedure (defined in \ref{append:OSmovingknife}) depends on the tie-breaking rule used.
Recall that, at Step~$t$ when the remaining part of the cake is allocated among the remaining $n-t+1$ agents, the algorithm computes a point $x_i$ for each agent $i$ where the interval from the left endpoint to $x_i$ is worth exactly $\frac1{n-t+1}$ fraction of the value of the remaining part of the cake. 
In the case agents are not necessarily hungry, the point $x_i$ may not be unique.
Same as before, we consider two natural tie-breaking rules: always choose the smallest such $x_i$ (OS-MovingKnife(left)) and always choose the largest such $x_i$ (OS-MovingKnife(right)).
We will show that the former is not risk-averse truthful and the latter is risk-averse truthful.

\begin{theorem}\label{thm:OS_RAT_left}
OS-MovingKnife(left) is not risk-averse truthful.
\end{theorem}
\begin{proof}
Notice that, when the total number of agents is $3$, OS-MovingKnife(left) is identical to Even-Paz($\lfloor\frac k2\rfloor$,left).
The proof of this theorem follows from the result in \ref{append:Even-Paz-not-RAT}.
%
%
%
\end{proof}

Before proving OS-MovingKnife(right) is risk-averse truthful, we first show the following simple property of Ortega and Segal-Halevi's moving-knife procedure.

\begin{proposition}\label{prop:OSMovingKnife}
Suppose, after an iteration, the remaining part of the cake is $[s,1]$ and the set of agents who have not yet been allocated is $S$.
Each agent $i\in S$ will be allocated a piece with value at least $\frac1{|S|}v_i([s,1])$ at the end of the procedure.
\end{proposition}
\begin{proof}
Let $S=\{1,\ldots,k\}$ where $i$ is the agent who is allocated at the $(T+i)$-th iteration, where $T$ is the iteration after which the remaining part of the cake is $[s,1]$.
Let $s_0,s_1,\ldots,s_k$ with $s=s_0<s_1<s_2<\cdots<s_k=1$ be such that $[s_{i-1},s_{i})$ is the piece allocated to agent $i$.
Consider an arbitrary agent $i$.
By the property of the procedure that the agent with the left-most mark point is allocated at each iteration, we have $\frac{v_i([s_{j-1},s_j))}{v_i([s_{j-1},1))}\leq\frac1{k-j+1}$ for each $j=1,\ldots,i-1$.
Thus, $\frac{v_i([s_{j},1])}{v_i([s_{j-1},1))}\geq\frac{k-j}{k-j+1}$ for each $j=1,\ldots,i-1$.
By multiplying these fractions for $j=1,\ldots,i-1$, we have $\frac{v_i([s_{i-1},1])}{v_i([s,1))}\geq\frac{k-i+1}k$.
Since the value of the piece $[s_{i-1},s_i)$ received by agent $i$ is exactly $\frac1{k-i+1}v_i([s_{i-1},1])$, we have $v_i([s_{i-1},s_i))=\frac1{k-i+1}v_i([s_{i-1},1])\geq\frac1{k-i+1}\cdot\frac{k-i+1}kv_i([s,1))=\frac1kv_i([s,1))$, which implies the proposition.
\end{proof}

Now we show the risk-averse truthfulness of OS-MovingKnife(right).

\begin{theorem}\label{thm:OS_RAT}
OS-MovingKnife(right) is risk-averse truthful.
\end{theorem}
\begin{proof}
Let $f_1$ be the true value density function for agent $1$. 
If agent $1$ misreports his/her value density function to $f_1'$, for any value density functions $f_2,\ldots,f_n$ of other agents, we consider the implementation of OS-MovingKnife(right) under $f_1,f_2 \ldots, f_n$ and $f_1',f_2 \ldots, f_n$.
Let $r$ and $r'$ be the iterations where agent $1$ is allocated a piece of cake at the two implementations respectively.
Let $r_m=\min\{r,r'\}$.
The two implementations of OS-MovingKnife(right) are identical for the first $r_m-1$ iterations.
Let $[s,1]$ be the part of the unallocated cake before the $r_m$-th iteration.
Let $S$ be the set of agents who have not been allocated before the $r_m$-th iteration, where we must have $1\in S$.
We discuss three cases: $r_m=r'<r$, $r_m=r'=r$, and $r'>r=r_m$.

If $r_m=r'<r$, agent $1$ is allocated at an earlier iteration by misreporting.
Recall that, at the $r_m$-th iteration, a point $x_i$ is computed for each agent $i$ such that $v_i([s,x_i])=\frac1{n-r_m+1}v_i([s,1])$, and the agent with the left-most point is allocated at this iteration.
Since agent $1$ is allocated at the $r_m$-th iteration by reporting $f_1'$ and is allocated at a later iteration by reporting $f_1$, the $\frac1{n-r_m+1}$-point $x_1$ with respect to $f_1$ must be larger than the $\frac1{n-r_m+1}$-point $x_1'$ with respect to $f_1'$.
By reporting $f_1'$, agent $1$ receives $[s,x_1')$ with value $v_1([s,x_1'))\leq v_1([s,x_1))=\frac1{n-r_m+1}v_1([s,1])$.
On the other hand, by Proposition~\ref{prop:OSMovingKnife}, agent $1$ receives value at least $\frac1{n-r_m+1}v_1([s,1])$ if reporting truthfully.
This shows that reporting $f_1'$ is not beneficial.

If $r_m=r'=r$, agent $1$ is allocated at the same iteration for reporting $f_1$ and $f_1'$.
Let $x_1$ and $x_1'$ be agent $1$'s $\frac1{n-r_m+1}$-points with respect to $f_1$ and $f_1'$ respectively.
If $x_1'\leq x_1$, agent $1$ receives only a subset by misreporting, which is not beneficial.
If $x_1'>x_1$, we have $v_1([s,x_1'))>v_1([s,x_1))$ since OS-MovingKnife(right) breaks ties at the right-most mark point.
In this case, misreporting $f_1'$ is beneficial for agent $1$.
However, we will construct $f_2',\ldots,f_n'$ such that this misreporting is harmful for agent $1$.
Suppose $f_2,\ldots,f_n,f_2',\ldots,f_n'$ are normalized such that the value of $[0,1]$ is always $1$.
Let $f_i$ and $f_i'$ be identical on $[0,s]$ for each $i=2,\ldots,n$, so the first $r_m-1$ iterations are the same as before.
Since $v_1([s,x_1'))>v_1([s,x_1))=\frac1{n-r_m+1}v_i([s,1])$, we have $v_1([x_1',1])<\frac{n-r_m}{n-r_m+1}v_i([s,1])$.
Find a point $x<x_1'$ such that $v_1([x,1])<\frac{n-r_m}{n-r_m+1}v_i([s,1])$.
We can construct $f_2',\ldots,f_n'$ by modifying the values of $f_1,\ldots,f_n$ on $[s,1]$ such that 1) an agent other than agent $1$ is allocated $[s,x)$ at the $r_m$-th iteration, 2) at each later iteration, there is an agent other than agent $1$ whose mark point is to the left of but very close to agent $1$'s mark point with respect to $f_1$, and 3) at the last iteration, the remaining part of the cake is worth less than $\frac{1}{n-r_m+1}v_i([s,1])$ to agent $1$.
It is then easy to see that agent $1$ will get a piece with value less than $\frac{1}{n-r_m+1}v_i([s,1])$ by reporting $f_1'$.
Since agent $1$ receives value $\frac1{n-r_m+1}v_i([s,1])$ by truthfully reporting $f_1$, misreporting is harmful in this case.

If $r'>r=r_m$, agent $1$ is allocated at a later iteration by misreporting.
Let $x_1$ and $x_1'$ be agent $1$'s $\frac1{n-r_m+1}$-points with respect to $f_1$ and $f_1'$ respectively at the $r_m$-th iteration.
It must be that $x_1'>x_1$.
Following the similar arguments in the last paragraph, we can construct $f_2',\ldots,f_n'$ such that reporting $f_1'$ is harmful for agent $1$.
\end{proof}

Again, the left-variant and the right-variant are identical if we assume agents are hungry.
\begin{corollary}
Ortega and Segal-Halevi's Moving-Knife Procedure is risk-averse truthful if agents are hungry.
\end{corollary}

\end{document}